\documentclass[11pt,a4paper]{article}

\usepackage{graphicx}
\usepackage{tikz-cd} 
\usepackage{amsmath, amsfonts, amsthm}
\usepackage{parskip}
\usepackage[T1]{fontenc}
\usepackage[margin=1in]{geometry}
\usepackage[
  style=numeric-comp,
  giveninits=true,
  maxbibnames=10,
  eprint=true,
  url=false,
  isbn=false,
  backend=biber,
  bibencoding=utf8,
  texencoding=ascii
  ]{biblatex}
\usepackage{booktabs}
\usepackage{array}
\usepackage{enumitem}
\usepackage{xcolor}
\usepackage[noblocks]{authblk}
\usepackage[hang, flushmargin]{footmisc}

\usepackage[colorlinks=true,linkcolor=red,citecolor=blue]{hyperref}

\theoremstyle{plain}
\newtheorem{theorem}{Theorem}
\newtheorem{lemma}[theorem]{Lemma}
\newtheorem{proposition}[theorem]{Proposition}
\theoremstyle{definition}
\newtheorem{definition}[theorem]{Definition}
\newtheorem{remark}[theorem]{Remark}
\newtheorem*{conjecture*}{Conjecture}

\newcommand{\email}[1]{\href{mailto:#1}{#1}}

\newcommand{\set}[1]{\mathbf{#1}}
\newcommand{\R}{\set{R}}
\newcommand{\C}{\set{C}}
\newcommand{\T}{\set{T}}
\newcommand{\Z}{\set{Z}}
\newcommand{\CP}{\mathbf{CP}}
\renewcommand{\S}{\set{S}}

\newcommand{\vv}[1]{\boldsymbol{#1}}
\newcommand{\group}[1]{\mathrm{#1}}

\DeclareMathOperator{\rank}{rank}
\newcommand{\so}{\operatorname{so}}
\newcommand{\symp}{\operatorname{sp}}

\newcommand{\SL}{\operatorname{SL}}

\renewcommand{\Re}{\operatorname{Re}}
\renewcommand{\Im}{\operatorname{Im}}

\newcommand{\ff}{\mathrm{f{}f}}
\newcommand{\fp}[2]{\mathrm{fp}_{#1,#2}}

\title{Hamiltonian Monodromy in a Tavis-Cummings System\\ with an $A_2$ Singularity}

\author[1]{Konstantinos Efstathiou\thanks{\email {k.efstathiou@dukekunshan.edu.cn}}}
\author[2]{Gabriela Jocelyn Gutierrez-Guillen\thanks{\email{gabriela.gutierrez@ens-lyon.fr}}}
\author[3,4,5]{Pavao Mardešić\thanks{\email{pavao.mardesic@ube.fr}}}
\author[6]{Dominique Sugny\thanks{\email{dominique.sugny@u-bourgogne.fr}}}

\affil[1]{Division of Natural and Applied Sciences and Zu Chongzhi Center, Duke Kunshan University, 8 Duke Avenue, 215316 Kunshan, Jiangsu, China}

\affil[2]{Unité de Mathématiques Pures et Appliquées UMR 5669 CNRS, École Normale Supérieure de Lyon Site Monod, 46 Allée d'Italie, 69364 Lyon Cedex 07, France}

\affil[3]{Institut de Mathématiques de Bourgogne UMR 5584 CNRS, Université Bourgogne Europe, 9 avenue Alain Savary, BP 47870, 21078 Dijon Cedex, France}

\affil[4]{Department of Mathematics, Faculty of Science, University of Zagreb, Bijeni\v{c}ka 30, 10 000 Zagreb, Croatia}

\affil[5]{Laboratorio Internacional Solomon Lefschetz IRL 2001 CNRS-UNAM, Instituto de Matemáticas Unidad Cuernavaca, Mexico}

\affil[6]{Universit{\'e} Bourgogne Europe, CNRS, Laboratoire Interdisciplinaire Carnot de Bourgogne ICB UMR 6303, 21000 Dijon, France}

\date{April 16, 2026}

\begin{document}

\maketitle

\begin{abstract}
Singular Lagrangian fibrations arising from three-degree-of-freedom integrable Hamiltonian systems remain largely unexplored. 
While several results describe the global structure of large classes of systems with two degrees of freedom, only a few examples are understood in higher dimensions.
We present a three-degree-of-freedom system derived from the two-spin Tavis-Cummings model whose singular Lagrangian fibration exhibits a topology that has not been observed in other physical models.
We show that the most degenerate singular fiber is homeomorphic to $\S^2\times\S^1$ with a singularity of $A_2$ type.
We further describe the bifurcation diagram and the global topology of the fibration, and we compute its Hamiltonian monodromy.
\end{abstract}


\begin{small}
\end{small}

\section{Introduction}

The Tavis-Cummings (TC) system is a fundamental model in quantum optics of the interaction between matter and electromagnetic fields ~\cite{TC1968}. 
The model describes the coupling between $N$ two-level atoms and a single-mode quantized radiation field within a lossless cavity~\cite{Haroche2006}.
It is itself a generalization of the Jaynes-Cummings (JC) system which applies to the case of a single atom, $N=1$~\cite{JC1963}.
The JC and TC systems have been extensively studied from classical and quantum points of view.
Notably, from the quantum point of view, the TC system has been realized experimentally in a superconducting circuit for quantum information processing~\cite{Majer2007, Fink2009}. 

In this work we consider the classical analogue of the quantum TC system, where the two-level quantum atoms are replaced by classical spins and the quantized cavity mode is replaced by a harmonic oscillator.
The classical TC system with $N$ spins is an $(N+1)$-degree-of-freedom Liouville integrable Hamiltonian system which, moreover, admits a Lax pair representation for any $N$~\cite{Lax1968, Babelon2003}.
The Lax pair considerably simplifies the analysis of the TC system, and leads to explicit solutions of the equations of motion in terms of hyperelliptic functions~\cite{Yuzbashyan2005a}.
Despite this explicit description of the TC dynamics, the qualitative understanding of the global topology of the system with $N \ge 2$ spins is still limited.

The JC system, corresponding to $N=1$, is well understood.
This is a two-degree-of-freedom (2-DOF) integrable Hamiltonian system with a global circle action, and it has a focus-focus singularity for a range of parameters.
The existence of a focus-focus singularity in a 2-DOF Hamiltonian system implies the non-triviality of Hamiltonian monodromy~\cite{Duistermaat1980, Vu-Ngoc1999, Zung1997}, 
and we note that the singular Lagrangian fibration near a focus-focus singularity is locally isomorphic to the $A_1$ singularity~\cite{Bates2007, Arnoldbook}.
Because of the focus-focus singularity, the JC system serves as one of the main examples in the study of semitoric systems and their symplectic classification \cite{Pelayo2009a}, and a complete description of its semitoric invariants is known~\cite{Pelayo2012, Alonso2019}.
The spectral Lax pair approach has also been used for the analysis of the global topology of the JC system~\cite{Babelon2012}, and to determine its Hamiltonian monodromy~\cite{Babelon2015, Gutierrez-Guillen2024}. 
Much less is known about the topology of the singular Lagrangian fibration of the TC system with $N \ge 2$.

This contrast between 2-DOF integrable Hamiltonian systems and systems with three or more degrees of freedom is not specific to the TC system. 
The topology of the singular Lagrangian fibrations of 2-DOF integrable Hamiltonian systems is well understood \cite{Bolsinov2004}, and there is an extensive collection of physical examples with two degrees of freedom manifesting different aspects of the mathematical theory, see \cite{Audin1996, Cushmanbook, Efstathioubook}.
However, very few studies have considered systems with three or more degrees of freedom and most of these have focused on 3-DOF systems with a global $\T^2$ action, where the analysis is facilitated by symplectic reduction to 1-DOF.
Examples include the Lagrange top~\cite{Audin1996, Cushmanbook, Dawson2022}, the resonant swing spring~\cite{Dullin:2004, Giacobbe:2004}, and the axially symmetric $1{:}1{:}{-}2$ resonance \cite{Efstathiou2019}.

Coming back to the TC system, and---in particular---the 3-DOF system with $N=2$ spins, we note that it possesses a global $\S^1$ action but lacks a global $\T^2$ action, and thus it cannot be symplectically reduced to 1-DOF.
For this system, the spectral Lax pair approach has been successfully used to parameterize its singularities and determine their stability type~\cite{Babelon2012}.
Nevertheless, the implications of these findings for the global topology of the two-spin TC system have not been further elaborated.

A natural question is whether among the qualitatively different types of singular Lagrangian fibrations that can appear in the TC system for different parameter choices, there are any novel cases that have not previously appeared in the literature.
A careful analysis of the spectral curve of the TC system, described in detail in Appendix~\ref{sec:lax}, allowed us to identify a choice of parameters for which the curve has a high-degree of degeneracy possessing two complex conjugate triple roots.
We call the two-spin TC system with this choice of parameters, the \emph{special Tavis-Cummings} (STC) system, see Definition~\ref{def:stcm}.

The STC system exhibits an intriguing global topology that has not been previously found in other physical examples.
We analyze in detail the singularities of the STC system and its Hamiltonian monodromy, and we show that:

\begin{enumerate}[label={(\roman*)},noitemsep]

\item The STC system has four threads of rank-$1$ focus-focus-regular singularities that meet at a degenerate critical value $\mathfrak c_*$, which we call the \emph{central critical value}, see Section~\ref{sec:critical-values} and Figure~\ref{fig:stcm-bd}.

\item The singular fiber of the integral map corresponding to the degenerate critical value $\mathfrak c_*$ is topologically equivalent to $\S^2 \times \S^1$ and contains an $\S^1$ orbit of degenerate singularities. We call the corresponding degenerate singularity $c_*$ in the reduced space, the \emph{central singularity}.

\item The $\S^1$ reduced fibration near the central singularity $c_*$ is locally isomorphic to the fibration of $f : \C^2 \times \R \to \C \times \R$, $f(x,y,\kappa) = (y^2 + x^3 + \kappa x, \kappa)$, near $0$, see Theorem~\ref{Th:local}. That is, the degenerate singularity in the reduced space is locally isomorphic to the $A_2$ singularity~\cite{Arnoldbook}.

\item The fundamental group of the set of regular values of the integral map is isomorphic to the free group of three generators $\mathbf{F}_3 = \Z \ast \Z \ast \Z$. For the choice of generators $\gamma_1$, $\gamma_2$, $\gamma_3$ described in Section~\ref{sec:monodromy-tavis-cummings} the corresponding  monodromy matrices are
\[
M(\gamma_1) = \begin{bmatrix} 1 & 0 & 0 \\ 0 & 1 & 1 \\0 & 0 & 1 \end{bmatrix}, \quad
M(\gamma_2) = \begin{bmatrix} 1 & 0 & 0 \\ 0 & 1 & 0 \\ 0 & -1 & 1  \end{bmatrix}, \quad
M(\gamma_3) = \begin{bmatrix} 1 & 0 & 0 \\ 0 & 2 & 1 \\ 0 & -1 & 0 \end{bmatrix}.
\]

\item For the value of the momentum generating the $\S^1$ action corresponding to the central singularity $c_*$, the symplectic $\S^1$ reduction of the STC system produces a 2-DOF integrable Hamiltonian system with an isolated critical value but no global $\S^1$ action, see Figure~\ref{fig:stcm-bd} with $K=k^*$. The singular fiber corresponding to the isolated critical value is homeomorphic to $\S^2$ and contains an $A_2$ singularity. The Hamiltonian monodromy around the isolated critical value is
\[ \begin{bmatrix} 1 & 1 \\ - 1 & 0 \end{bmatrix}. \]
\end{enumerate}

The STC system presents a novel and interesting physical example of a 3-DOF integrable Hamiltonian system with non-trivial Hamiltonian monodromy and the appearance of an $A_2$ singularity.
From this point of view, it is an important addition to our short list of well understood singular Lagrangian fibrations of 3-DOF integrable Hamiltonian systems.
Such examples will play an important role in any future attempts toward the topological and symplectic classification of 3-DOF integrable Hamiltonian systems. 
Moreover, such examples are useful in the context of mirror symmetry \cite{Gross2001, Bernard2009}.

The paper is organized as follows. 
In Section~\ref{sec:definitions} we present the Hamiltonian and the integral map of the classical two-spin TC system and the special two-spin TC (STC) system. 
In Section~\ref{sec:reduction-s1-action} we discuss the $\S^1$ action of the TC system, its fixed points, and local symplectic reduction.
In Section~\ref{sec:critical-values} we parameterize and classify the singularities of the TC system, obtaining Figure~\ref{fig:stcm-bd}. 
We note that even though we do not use the spectral Lax pair approach, we recover the results of~\cite{Babelon2012} for rank-$0$ and rank-$1$ singularities, while slightly improving on the previous results for rank-$2$ singularities which do not provide explicit parameter ranges.
In Section~\ref{sec:global-topology}, we numerically calculate the Hamiltonian monodromy of the STC system. 
In Section~\ref{sec:a2-singularity} we show that the STC system in the vicinity of the central singularity has a local normal form corresponding to an $A_2$ singularity, that is the fibration of the STC system is locally topologically equivalent to the fibration given by the versal unfolding of the $A_2$ singularity.
We then compute the Picard-Lefschetz monodromy of the normal form and we show its correspondence with the Hamiltonian monodromy of the system.
We give conclusions and prospective views in Section~\ref{sec:conclusions}.
Finally, in the Appendices, we provide additional information and details on the reduction of the $\S^1$ action, the proof of the parameterization of rank-$2$ singularities, the numerical computation of Hamiltonian monodromy, and a brief description of the spectral Lax pair formalism.

\section{The classical two-spin Tavis-Cummings system}
\label{sec:definitions}

The classical two-spin Tavis-Cummings system is a 3-DOF Hamiltonian system defined on the symplectic manifold $M = \S^2_u \times \S^2_v \times \R^2$. 
The symplectic form on each of the spheres 
\[ 
\S^2_u = \{ (u_1,u_2,u_3) : u_1^2 + u_2^2 + u_3^2 = 1 \},
\ 
\S^2_v = \{ (v_1,v_2,v_3) : v_1^2 + v_2^2 + v_3^2 = 1 \},
\] 
is the standard area form, while $\R^2$ has canonical coordinates $(q,p)$.
The corresponding Poisson algebra is isomorphic to $\so(3) \oplus \so(3) \oplus \symp(2,\R)$, where the non-vanishing Poisson brackets are given by
\begin{align*}
\{ u_i, u_j \} = \sum_k \varepsilon_{ijk} u_k, \ 
\{ v_i, v_j \} = \sum_k \varepsilon_{ijk} v_k, \ 
\{ q, p \} = - \{ p, q \} = 1.
\end{align*}

The Hamiltonian function of the two-spin TC system is
\begin{align*}
H = \delta_1 u_3 + \delta_2 v_3
+ \frac{\omega}{2} (p^2 + q^2)
+ \sqrt2 g (q u_1 - p u_2) + \sqrt2 g (q v_1 - p v_2),
\end{align*}
where $\delta_1$, $\delta_2$, $\omega$, $g$ are real parameters.

We introduce the complex coordinates
\begin{align*}
u = u_1 + i u_2,\ v = v_1 + i v_2, \ z = p + i q,
\end{align*}
with Poisson brackets
\begin{gather*}
\{ u, u_3 \} = i u, \ \{ \bar{u}, u_3 \} = - i \bar{u}, \ \{ u, \bar{u} \} = -2 i u_3, \\
\{ v, v_3 \} = i v, \ \{ \bar{v}, v_3 \} = - i \bar{v}, \ \{ v, \bar{v} \} = -2 i v_3,
\{ z, \bar{z} \} =  2i. 
\end{gather*}
Using these coordinates, $H$ is written as
\begin{align*}
H = \delta_1 u_3 + \delta_2 v_3
+ \frac{\omega}{2} z \bar z
+ \frac{i}{\sqrt2} g (u \bar z - \bar u z) 
+ \frac{i}{\sqrt2} g (v \bar z - \bar v z),
\end{align*}

The function $H$ is invariant under the Hamiltonian $\S^1$ action
\begin{align}
\label{eq:S1-action}
\varphi_K: \S^1 \times M \to M : (e^{it}, (u, u_3; v, v_3; z)) 
\mapsto (e^{it} u, u_3; e^{it} v, v_3; e^{it} z),
\end{align}
generated by the Hamiltonian flow of the momentum
\begin{equation}
\label{eq:K}
K = u_3 + v_3 + \frac12 z \bar z,
\end{equation}
with Hamiltonian vector field $X_K$ satisfying
\begin{equation}
\label{eq:vf-X_K}
X_K(u) = i u,\  X_K(u_3) = 0,\ X_K(v) = i v, \ X_K(v_3) = 0, \ X_K(z) = i z.
\end{equation}

Although there is only one obvious symmetry in the system, it is well known that the Hamiltonian $H$ is Liouville integrable~\cite{Babelon2012}. 
In addition to the momentum $K$, two other integrals are given by
\begin{subequations}
\label{eq:H1-H2-complex}
\begin{align}
\label{eq:H1-complex}
H_1 &= (\delta_1-\omega) u_3 + \frac{i}{\sqrt2} g (u \bar{z} - \bar{u} z) 
- \frac{g^2}{\delta_2 - \delta_1} (u \bar{v} + \bar{u} v + 2 u_3 v_3), \\
\label{eq:H2-complex}
H_2 &= (\delta_2-\omega) v_3 + \frac{i}{\sqrt2} g (v \bar{z} - \bar{v} z) 
+ \frac{g^2}{\delta_2 - \delta_1} (u \bar{v} + \bar{u} v + 2 u_3 v_3).
\end{align}
\end{subequations}
Then,
\[ H = H_1 + H_2 + \omega K. \]

Both $H_1$ and $H_2$ are invariant under the Hamiltonian $\S^1$ action \eqref{eq:S1-action}, which implies $\{H_1,K\} = \{H_2,K\} = 0$.
Moreover, a direct computation shows that $\{H_1,H_2\} = 0$. 
Therefore, the map
\begin{equation}\label{eq:ihs-tc}
F = (H_1,H_2,K) : M \to \R^3,
\end{equation}
defines a Liouville integrable system.

\begin{remark}
In the case $\delta_1=\delta_2$, the Hamiltonian dynamics is completely different. In this case, there are three independent integrals $H$, $K$, and $J = \sum_{k=1}^3 (u_k+v_k)^2$. The momentum $K$ generates a global $\S^1$ action, while $J$ generates a periodic flow and thus a global $\S^1$ action up to time rescaling.
\end{remark}

The following expressions for the Hamiltonian vector fields $X_{H_1}$, $X_{H_2}$ will be used in Section~\ref{sec:critical-values} for deriving the singularities of $F$:
\begin{equation}
\label{eq:vf-X_H1}
\begin{aligned}
& X_{H_1}(u) = (\delta_1-\omega) i u - \sqrt2 g u_3 z + \frac{2ig^2}{\delta_2 - \delta_1} (u_3 v - v_3 u), 
&& X_{H_1}(v) = - \frac{2 i g^2}{\delta_2 - \delta_1} (u_3 v - v_3 u), \\
& X_{H_1}(u_3) = \frac{g}{\sqrt2} (u \bar{z} + \bar{u} z) + \frac{i g^2}{\delta_2 - \delta_1} (u \bar{v}-\bar{u} v), 
&& X_{H_1}(v_3) = - \frac{i g^2}{\delta_2 - \delta_1} (u \bar{v} - \bar{u} v), \\
& X_{H_1}(z) = - \sqrt2 g u,
\end{aligned}
\end{equation}
and
\begin{equation}
\label{eq:vf-X_H2}
\begin{aligned}
& X_{H_2}(u) = -\frac{2 i g^2}{\delta_2 - \delta_1} (u_3 v-u v_3), 
&& X_{H_2}(v) = (\delta_2-\omega) i v -\sqrt{2} g v_3 z + \frac{2 i g^2}{\delta_2 - \delta_1} (u_3 v - u v_3), \\
& X_{H_2}(u_3) = -\frac{i g^2}{\delta_2 - \delta_1} (u \bar{v}-\bar{u} v), 
&& X_{H_2}(v_3) = \frac{g}{\sqrt2} (v \bar{z}+\bar{v} z) + \frac{i g^2}{\delta_2 - \delta_1} (u \bar{v} - \bar{u} v), \\
& X_{H_2}(z) = -\sqrt2 g v.
\end{aligned}
\end{equation}

As noted in the Introduction, in this work we focus on a special case of the two-spin Tavis–Cummings system.
The corresponding parameter values were obtained using a spectral Lax pair approach~\cite{izosimov2017singularities}, where the problem of finding singularities is translated into the problem of finding multiple roots of a polynomial, whose vanishing locus is called the \emph{spectral curve}~\cite{Babelon2012}. 
For the two-spin TC system, by searching for parameters such that this degree $6$ polynomial exhibits two complex conjugate triple roots, we find that it is necessary to analyze the parameter values satisfying the resonance condition $\delta_1+\delta_2=2\omega$, with $\delta_1 \ne \delta_2$, as detailed in Appendix~\ref{sec:lax}. 
For computational simplicity, we further restrict attention to the following specific values of the TC parameters.

\begin{definition}
\label{def:stcm}
The \emph{special Tavis-Cummings} (STC) system is the two-spin Tavis-Cummings system described by the integral map $F = (H_1,H_2,K)$ with $H_1$, $H_2$ in Eq.~\eqref{eq:H1-H2-complex} and parameter values
\begin{equation}\label{eq:stcm}
\delta_1 = \frac12 ,\quad \delta_2 = \frac32, \quad \omega=1, \quad g=1.
\end{equation} 
\end{definition}

Note that several of the results in Section~\ref{sec:reduction-s1-action} and Section~\ref{sec:critical-values} apply to the general two-spin TC system, and we explicitly mention when a result applies only to the STC system. Results in Section~\ref{sec:global-topology} and Section~\ref{sec:a2-singularity} apply specifically to the STC system.

\section{The \texorpdfstring{$\S^1$}{S1} symmetry of the Tavis-Cummings system}
\label{sec:reduction-s1-action}

In this section, we consider the $\S^1$ action $\varphi_K$ \eqref{eq:S1-action} of the Tavis-Cummings system, we determine the fixed points of the action, and we describe how to locally reduce the $\S^1$ action. 
We refer the interested reader to Appendix~\ref{sec:algebraic-reduction} for a discussion of the global reduction of the $\S^1$ action using algebraic invariants, and a description of the topology of the reduced spaces $K^{-1}(k) / \S^1$.

\begin{proposition}
\label{prop:S1-fixed-points}
The $\S^1$ action $\varphi_K$ \eqref{eq:S1-action} on $M = \S^2_u \times \S^2_v \times \C$ is free outside the fixed points $\fp{\sigma_u}{\sigma_v} = (0,\sigma_u;0,\sigma_v;0)$ where $\sigma_u,\sigma_v \in \{-1,+1\}$.
\end{proposition}

\begin{proof}
Equation~\eqref{eq:S1-action} shows that the only points with non-trivial isotropy are the points with $u=v=z=0$.
Then, the relations $u \bar{u} + u_3^2 = 1$ and $v \bar{v} + v_3^2 = 1$ imply that $u_3=\pm 1$ and $v_3=\pm 1$.
\end{proof}

The four fixed points given in Proposition~\ref{prop:S1-fixed-points} satisfy
\begin{equation}
\label{eq:rank-0-critical-values}
F(\fp{\sigma_u}{\sigma_v}) = \Bigl(
(\delta_1-\omega) \sigma_u  - \frac{2 g^2}{\delta_2 - \delta_1} \sigma_u \sigma_v,
(\delta_2-\omega) \sigma_v  + \frac{2 g^2}{\delta_2 - \delta_1} \sigma_u \sigma_v,
\sigma_u + \sigma_v
\Bigr).
\end{equation}

Notice that for $k \not\in \{-2,0,2\}$, the set $K^{-1}(k)$ does not contain any fixed points of the $\S^1$ action $\varphi_K$.
This implies that the reduced space $K^{-1}(k) / \S^1$ is a smooth manifold for $k \in (-2,0) \cup (0,2) \cup (2,\infty)$.
The topology of these reduced spaces is given in Proposition~\ref{prop:reduced-spaces}.

It will be useful in the following to reduce the $\S^1$ action $\varphi_K$ and obtain a two-degree-of-freedom Hamiltonian system on the reduced space $K^{-1}(k) / \S^1$.
The reduction of the $\S^1$ action can be performed globally using algebraic invariants, that is, using the Hilbert basis for the given group action.
Reduction using algebraic invariants is described in detail in Appendix~\ref{sec:algebraic-reduction}, where it is shown that it leads to a description of the four-dimensional reduced space $K^{-1}(k) / \S^1$ with $9$ invariants that satisfy a number of syzygies which are not all independent, and with a non-standard Poisson structure.
From this point of view, even though algebraic reduction has the benefit of being globally applicable, it can be cumbersome.
Therefore, if we are interested only in a local description of the dynamics, e.g., near a relative equilibrium, it is more convenient to reduce the $\S^1$ action so that we are left with exactly $4$ coordinates in the reduced space.

To this effect, we first introduce canonical coordinates $(\theta_u, u_3; \theta_v, v_3; q, p)$ on $\S^2_{*,u} \times \S^2_{*,v} \times \R^2$, where $\S^2_{*,u,v} = \S^2_{u,v} \setminus \{(0,0,\pm1)\}$, with 
\[ u = u_1 + i u_2 = (1-u_3^2)^{1/2} \exp(i \theta_u), \ 
   v = v_1 + i v_2 = (1-v_3^2)^{1/2} \exp(i \theta_v). 
\]
The given coordinates are canonical with $\{\theta_u,u_3\} = \{\theta_v,v_3\} = \{q,p\} = 1$.

Let $M_* = \S^2_{*,u} \times \S^2_{*,v} \times \R^2_*$, where $\R^2_* = \R^2 \setminus \{(0,0)\}$.
Then, the $\S^1$ action $\varphi_K$ on $K^{-1}(k) \cap M_*$ is reduced by defining the section 
\[ s : \S^2_{*,u} \times \S^2_{*,v} \to K^{-1}(k) \cap M_* : (\theta_u, u_3; \theta_v, v_3) 
\mapsto (\theta_u, u_3; \theta_v, v_3; [2(k-u_3-v_3)]^{1/2}, 0). \]
For any smooth $\S^1$ invariant function $H$ on $M_*$ the corresponding reduced function on $\S^2_{*,u} \times \S^2_{*,v}$ is given by $\widehat H = s^* H$.
The functions $H_1$, $H_2$ in \eqref{eq:H1-H2-complex} reduce to
\begin{subequations}
\label{eq:H1-H2-reduced}
\begin{align}
& \begin{aligned}
\widehat H_1 &= (\delta_1-\omega) u_3 
+ 2 g (1-u_3^2)^{1/2} (k-u_3-v_3)^{1/2} \cos\theta_u \\
& \quad - \frac{2g^2}{\delta_2 - \delta_1} \bigl[ (1-u_3^2)^{1/2} (1-v_3^2)^{1/2} \cos(\theta_u-\theta_v)
+ u_3 v_3 \bigr], 
\end{aligned}
\\
& \begin{aligned}
\widehat H_2 &= (\delta_2-\omega) v_3
+ 2 g (1-v_3^2)^{1/2} (k-u_3-v_3)^{1/2} \cos\theta_v \\
& \quad + \frac{2g^2}{\delta_2 - \delta_1} \bigl[ (1-u_3^2)^{1/2} (1-v_3^2)^{1/2} \cos(\theta_u-\theta_v)
+ u_3 v_3 \bigr].
\end{aligned}
\end{align}
\end{subequations}

\section{Singularities of the Tavis-Cummings system}
\label{sec:critical-values}

In this section, we parameterize the singularities of the integral map $F$ of the TC system and we determine their type for the STC system. 
Rank-$0$ singularities are discussed in Section~\ref{sec:rank-0-singularities}, while rank-$1$ and rank-$2$ singularities are obtained in Sections~\ref{sec:rank-1-singularities} and~\ref{sec:rank-2-singularities}, respectively. 
We summarize here the main results of this section for the STC system and we depict its set of critical values in Figure~\ref{fig:stcm-bd}.

\begin{figure}
\centering
\includegraphics[width=0.62\linewidth]{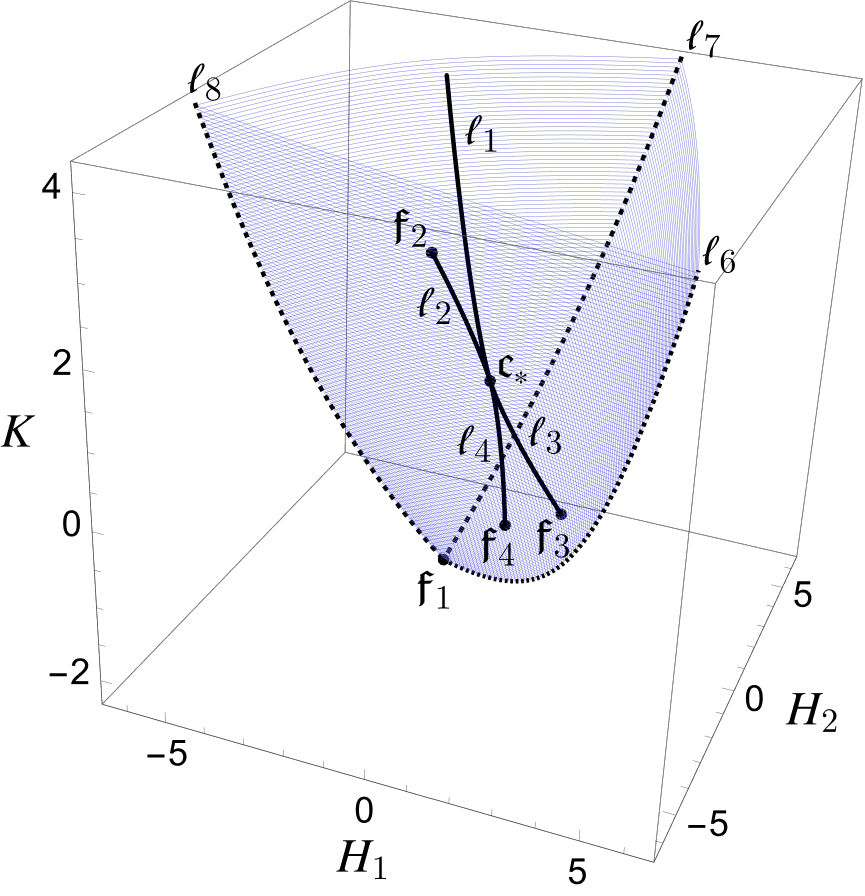}
\\[1cm]
\includegraphics[width=0.32\linewidth]{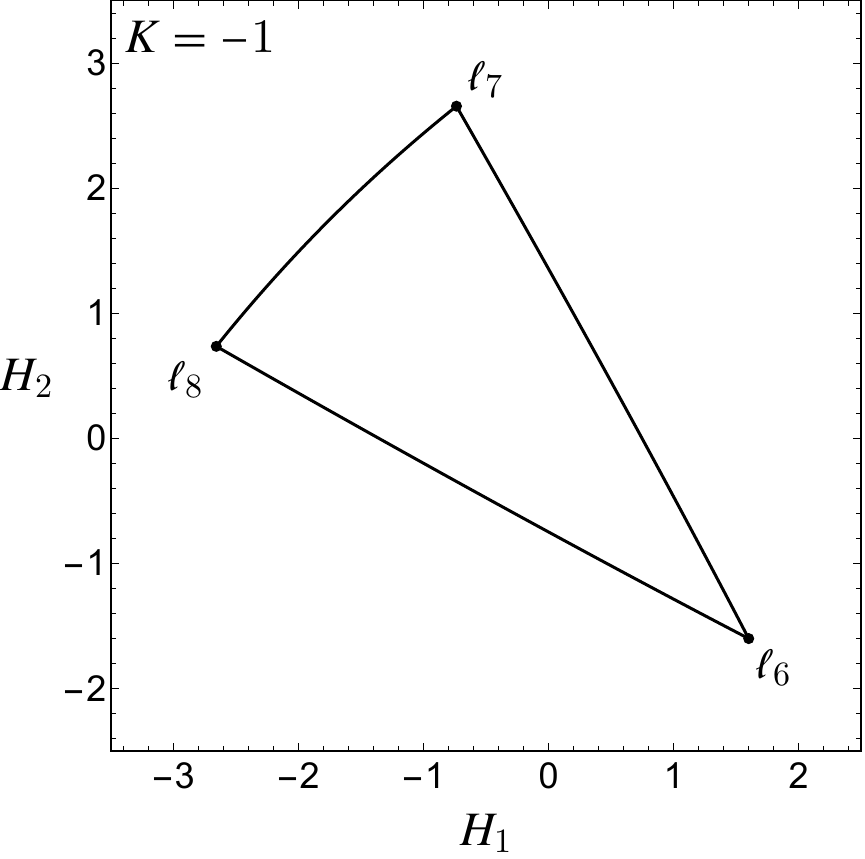} \hfill
\includegraphics[width=0.32\linewidth]{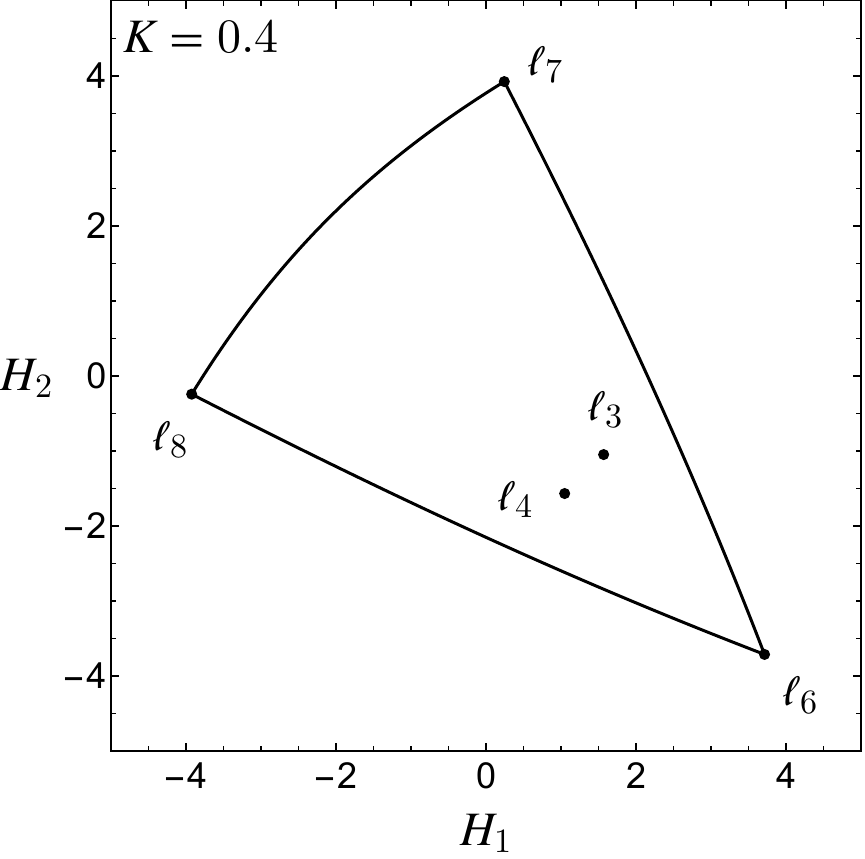} \hfill
\includegraphics[width=0.32\linewidth]{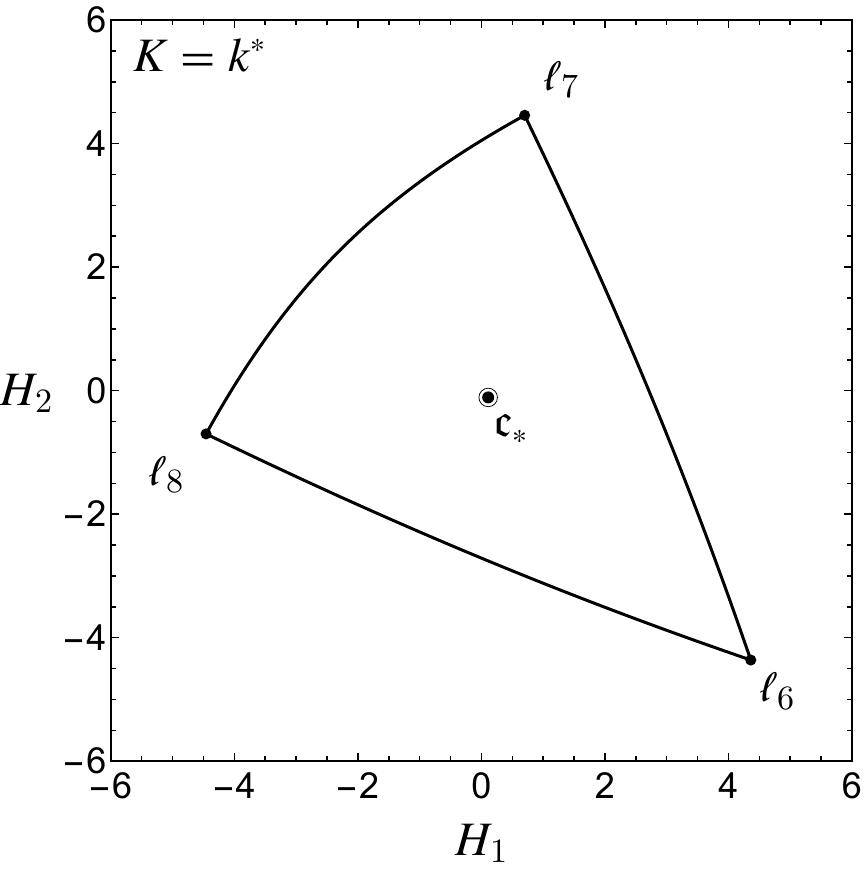} 
\\
\includegraphics[width=0.32\linewidth]{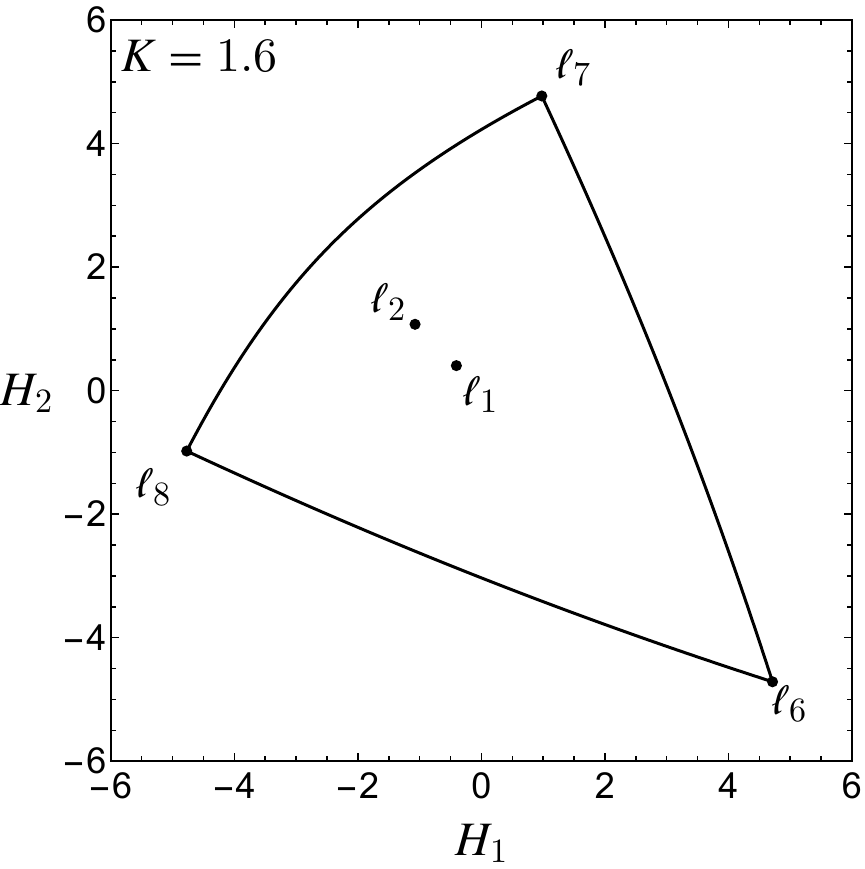} \hfill
\includegraphics[width=0.32\linewidth]{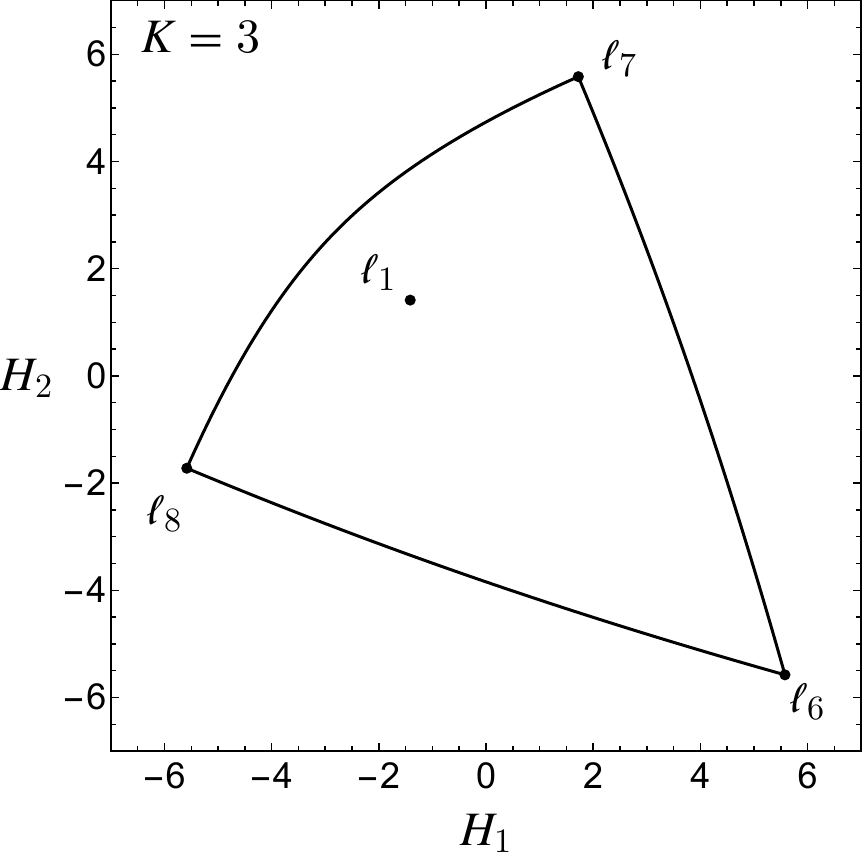} \hfill
\includegraphics[width=0.32\linewidth]{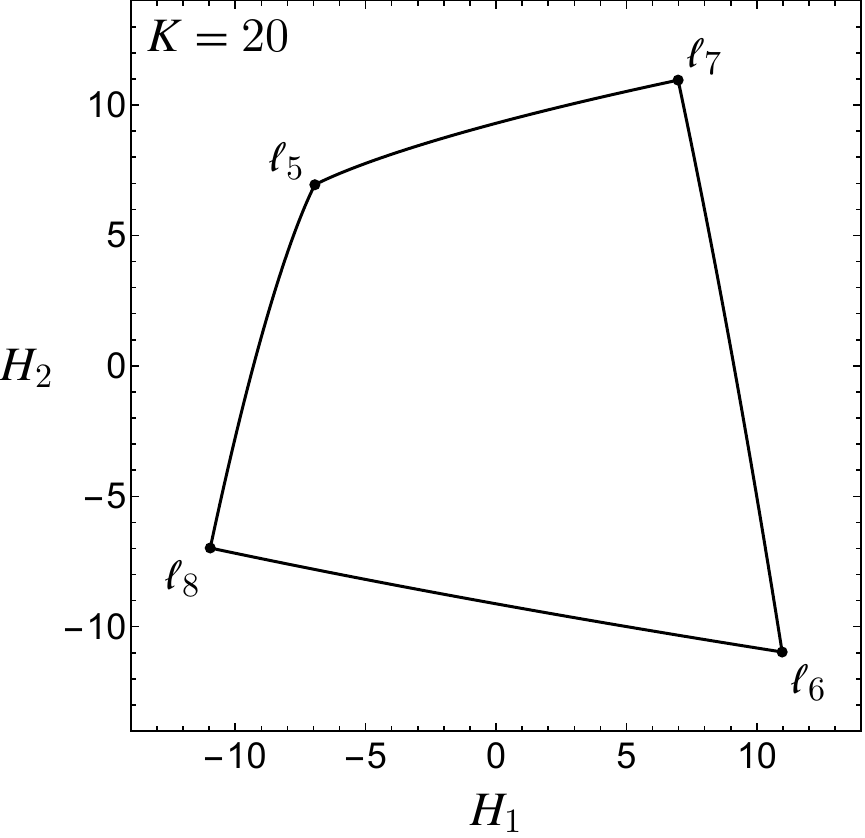}
\caption{Critical values of the STC system. Top: Critical values are depicted in the space $(H_1,H_2,K)$. Solid black lines represent the rank-$1$ FFR families $\ell_1$, $\ell_2$, $\ell_3$, $\ell_4$; they meet at the central critical value $\mathfrak c_*$. Dashed black lines represent the EER families $\ell_6$, $\ell_7$, $\ell_8$. The EER family $\ell_5$ appears for $K \ge 16.0625$ and is not shown here. The thin blue lines represent rank-$2$ singularities. Bottom: Horizontal slices of the set of critical with planes $K=k$.}
\label{fig:stcm-bd}
\end{figure}

The integral map $F=(H_1,H_2,K)$ of the STC system has four rank-$0$ critical values which are given by Eq.~\eqref{eq:rank-0-critical-values} and lie on the boundary of the image of $F$. We denote these by
\begin{gather*}
\mathfrak f_1 = F(\fp{-1}{-1}) = (-\tfrac32,\tfrac32,-2), \quad
\mathfrak f_2 = F(\fp{+1}{+1}) = (-\tfrac52,\tfrac52,2), \\
\mathfrak f_3 = F(\fp{-1}{+1}) = (\tfrac52,-\tfrac32,0), \quad 
\mathfrak f_4 = F(\fp{+1}{-1}) = (\tfrac32,-\tfrac52,0).
\end{gather*}

The rank-$1$ focus-focus-regular critical values form four threads $\ell_j$, $j=1,\dots,4$ in the interior of the image of $F$ that meet at a degenerate rank-$1$ critical value
\[ \mathfrak c_* = (h_1^*, h_2^*, k^*) 
= \Bigl( 2 - 3 \cdot 2^{2/3} , 
  -2 + 3 \cdot 2^{2/3} , 
  \frac{1}{16} (12 \cdot 2^{2/3}  - 1) \Bigr), \]
that we call the \emph{central critical value}. Each one of the threads $\ell_j$, $j=2,3,4$, connects the boundary point $\mathfrak f_j$ to the central critical value $\mathfrak c_*$. The thread $\ell_1$ starts at $\mathfrak c_*$ and meets the boundary of the image of $F$ at a point with $K=16.0625$, corresponding to a supercritical Hamiltonian Hopf bifurcation \cite{Meer1985}, and then continues along the boundary as the elliptic-elliptic-regular thread $\ell_5$. Each critical value $c$ along the threads $\ell_j$, $j=1,\dots,4$, has focus-focus-regular type, and thus the corresponding fiber $F^{-1}(c)$ is homeomorphic to the product of a two-dimensional pinched torus and $\S^1$. The central critical value corresponds to a fiber $F^{-1}(\mathfrak c_*)$ which is homeomorphic to $\S^2 \times \S^1$ and contains an $\S^1$ orbit of degenerate singularities, see Proposition~\ref{prop:a2-sphere}.

The boundary of the image of $F$ is the closure of the set of rank-$2$ critical values. Besides the rank-$2$ critical values, it contains the fixed points $\mathfrak f_j$, $j=1,\dots,4$, and the remaining four threads of rank-$1$ critical values of elliptic-elliptic-regular type denoted by $\ell_j$, $j=5,\dots,8$.

\subsection{Singularities in integrable Hamiltonian systems}
\label{sec:singulatities-ihs}

Consider an integrable Hamiltonian system $F = (F_1,\dots,F_n): M^{2n} \to \R^n$, and denote by $X_{F_j}$, $j=1,\dots,n$, the corresponding Hamiltonian vector fields.
A point $c \in M$ is a \emph{rank-$r$ singular point} (\emph{singularity}; \emph{critical point}) of $F$ if $\rank DF(c) = r < n$. 
The set $C$ of singular points of $F$ admits a natural stratification based on the rank.
In particular, $C$ is stratified by the disjoint sets $C_r = \{ c \in M : \rank DF(c) = r < n \}$ with $r=0,1,\dots,n-1$.

Define 
\[ L = \operatorname{span} \{ X_{F_1}(c), \dots, X_{F_n}(c) \} \subseteq T_c M. \]
Then, the condition that $c$ is a rank-$r$ singular point is equivalent to $\dim L = r < n$ and $L$ is an isotropic subspace of $T_c M$.
In this case, there are linear combinations $G_j(x) = \sum_{i=1}^n k_i (F_i(x) - F_i(c))$, $j=1,\dots,n$, with $k_i \in \R$, such that $G_j(c)=0$ for $j=1,\dots,n$, $X_{G_1}(c)=\cdots=X_{G_{n-r}}(c)=0$, and $X_{G_j}(c) \ne 0$ for $j=n-r+1,\dots,n$.

A rank-$r$ singular point $c \in M$ is \emph{non-degenerate} if for some $\lambda_1, \dots, \lambda_{n-r} \in \R$ the linear combination
\[ \lambda_1 DX_{G_1}(c) + \dots + \lambda_{n-r} DX_{G_{n-r}}(c) \bigl|_{L^\omega/L}, \]
has $2(n-r)$ distinct non-zero eigenvalues, where $DX_{G_j}(c)$ is the linearization at $c$ of the vector field $X_{G_j}$, and $L^\omega \supseteq L$ is the symplectic orthogonal of $L$, cf.~\cite[Definition~1.24]{Bolsinov2004}.
A point $b \in \R^n$ is a \emph{critical value} of $F$ if $F^{-1}(b)$ contains any singular points of $F$.
The \emph{bifurcation diagram} $\mathcal{BD}$ associated to $F$ is the set of critical values of $F$, that is, $\mathcal{BD} = F(C)$.

In the case of the two-spin Tavis-Cummings system, where $n=3$, a singular point can have rank $r$ equal to $0$, $1$, or $2$. We consider these cases separately in the following subsections.

\subsection{Rank-0 singularities}
\label{sec:rank-0-singularities}

If $r = \rank DF(c) = 0$, then $X_{F_i}(c) = 0$, for $i=1,2,3$.
In this case, $L = \{0\}$ and $L^\omega / L = T_c M$.
Therefore, $c$ is a non-degenerate singular point, if for some $\lambda_1, \lambda_2, \lambda_3 \in \R$, the linearized vector field
\[ \lambda_1 DX_{F_1}(c) + \lambda_2 DX_{F_2}(c)
+ \lambda_3 DX_{F_3}(c) \]
has $6$ distinct non-zero eigenvalues.
This implies that the eigenvalues can combine in one of the following types, where $\alpha_j \in \R \setminus \{0\}$, and the eigenvalues are distinct:

\begin{description}[noitemsep,font={\normalfont}]
\item[EEE:] elliptic-elliptic-elliptic, with eigenvalues $\pm i \alpha_1$, $\pm i \alpha_2$, $\pm i \alpha_3$.
\item[EEH:] elliptic-elliptic-hyperbolic, with eigenvalues $\pm i \alpha_1$, $\pm i \alpha_2$, $\pm \alpha_3$.
\item[EHH:] elliptic-hyperbolic-hyperbolic, with eigenvalues $\pm i \alpha_1$, $\pm \alpha_2$, $\pm \alpha_3$.
\item[HHH:] hyperbolic-hyperbolic-hyperbolic, with eigenvalues $\pm \alpha_1$, $\pm \alpha_2$, $\pm \alpha_3$.
\item[EFF:] elliptic-focus-focus, with eigenvalues $\pm i \alpha_1$, $\pm \alpha_2 \pm i \alpha_3$.
\item[HFF:] hyperbolic-focus-focus, with eigenvalues $\pm \alpha_1$, $\pm \alpha_2 \pm i \alpha_3$.
\end{description}

In the Tavis-Cummings system, rank-$0$ singular points satisfy $X_K(c) = 0$. 
Therefore, such points are the fixed points of the $\S^1$ action which as we have seen in Section~\ref{sec:reduction-s1-action} are given by $(u,u_3;v,v_3;z)=(0,\pm1;0,\pm1;0)$.
The type of the rank-$0$ singularities of the STC system is given in the following result.

\begin{proposition}\label{Prop:rank0}
In the STC system, the rank-$0$ singularity $(0,-1;0,-1;0)$ is EEE (elliptic-elliptic-elliptic), while the rank-$0$ singularities $(0,1;0,-1;0)$, $(0,-1;0,1;0)$, and $(0,1;0,1;0)$ are EFF (elliptic-focus-focus).
\end{proposition}

\begin{proof}
It is sufficient to compute the eigenvalues of $DX_{H_1}(c)$.
We consider local charts $(u,v,z)$ near each of the four singular points $(0,s_1;0,s_2;0)$ with $s_1,s_2 \in \{-1,+1\}$. 
In each of these charts we have $u_3 = s_1 (1-u \bar{u})^{1/2}$ and $v_3 = s_2 (1-v \bar{v})^{1/2}$ and we compute in coordinates $(u, \bar{u}; v, \bar{v}; z, \bar{z})$ that
\[
DX_{H_1} = \begin{bmatrix}
-2 i s_2-\frac{i}{2} & 0 & 2 i s_1 & 0 & -\sqrt{2} s_1 & 0 \\
 0 & 2 i s_2+\frac{i}{2} & 0 & -2 i s_1 & 0 & -\sqrt{2} s_1 \\
 2 i s_2 & 0 & -2 i s_1 & 0 & 0 & 0 \\
 0 & -2 i s_2 & 0 & 2 i s_1 & 0 & 0 \\
 -\sqrt{2} & 0 & 0 & 0 & 0 & 0 \\
 0 & -\sqrt{2} & 0 & 0 & 0 & 0 \\
\end{bmatrix}.
\]
We find that for $s_1=s_2=-1$ the eigenvalues are $\pm 4i$, $\pm \frac14(\sqrt{17}+1) i$, $\pm \frac14(\sqrt{17}-1) i$,
and the singularity is EEE.
For $s_1=s_2=1$ the eigenvalues are
$\pm 4i$, $\frac{1}{4}(\pm \sqrt{15} \pm i)$,
and the singularity is EFF.
Finally, for $s_1=1$, $s_2=-1$ and for $s_1=-1$, $s_2=1$ the roots of the characteristic polynomial can be computed exactly but the resulting expressions are cumbersome. We give here the approximate values $\pm 1.095 i$, $\pm 1.89 \pm 0.298 i$ for $s_1=1$, $s_2=-1$, and $\pm 2.42 i$, $\pm 0.854 \pm 0.961 i$ for $s_1=-1$, $s_2=1$.
In both cases, the singularity has type EFF.
\end{proof}

\subsection{Rank-1 singularities}
\label{sec:rank-1-singularities}

We give a parameterization of rank-$1$ singularities of the general TC system.

\begin{proposition}
\label{prop0:rank-1-parameterization}
The rank-$1$ singularities of the Tavis-Cummings system are parameterized as 
\begin{equation}\label{eq:rank-1-phase-space-general}
\begin{gathered}
z \bar z = \frac{2 g^2}{x^2} - \frac{(x+y-(\delta_1-\omega))^2}{2 g^2}=\frac{2 g^2}{y^2} - \frac{(x+y-(\delta_2-\omega))^2}{2 g^2}, 
\qquad
u = \frac{xi}{-\sqrt2 g} z,
\quad
v = \frac{yi}{-\sqrt2 g} z,
\\
u_3 = - \frac{x}{2 g^2} (x + y - (\delta_1-\omega)), 
\quad 
v_3 = - \frac{y}{2 g^2} (x + y - (\delta_2-\omega)).
\end{gathered}
\end{equation}
The corresponding critical values are parameterized as
\begin{subequations}
\label{eq:rank-1-param-general}
\begin{align}
h^c &= - \frac{1}{2 g^2} ((x + y)^2 - x(\delta_1-\omega) -y(\delta_2-\omega)) 
+ \frac{g^2}{x^2} - \frac{(x+y-(\delta_1-\omega))^2}{4 g^2}, \\
h_1^c &= -\frac{2 g^2 (y-\delta_2+\delta_1)}{(\delta_2-\delta_1) x}-\frac{x^2 (x+y-(\delta_1-\omega) )}{2 g^2}, \\
h_2^c &= \frac{2 g^2 (x+\delta_2-\delta_1)}{(\delta_2-\delta_1) y} -\frac{y^2 (x+y-(\delta_2-\omega) )}{2 g^2}, 
\end{align}
\end{subequations}
where $x$, $y\in\R$ satisfy
\begin{equation}\label{eq:relation-xy}
  4 g^4 (x^2-y^2) + (\delta_2 - \delta_1) x^2 y^2 (2(x + y) - (\delta_1 + \delta_2 - 2 \omega)) = 0,
\end{equation}
and are constrained by the conditions $-1 \le u_3 \le 1$ and $-1 \le v_3 \le 1$.
\end{proposition}

\begin{proof}
If $c$ is a rank-$1$ singularity of $F$ then we must have $X_K(c) \ne 0$.
Therefore, to find rank-$1$ singularities we require that $X_{H_1} = x X_K$ and $X_{H_2} = y X_K$ for some $x, y \in \R$ and we make use of Eq.~\eqref{eq:vf-X_K}, Eq.~\eqref{eq:vf-X_H1}, and Eq.~\eqref{eq:vf-X_H2}. 

First, $X_{H_1}(z) = x X_K(z)$ and $X_{H_2}(z) = y X_K(z)$ give
\[ -\sqrt2 g u = x i z, \ -\sqrt2 g v = y i z. \]
Notice that these relations directly give $u \bar{v} - \bar{u} v = 0$, $u \bar{z} + \bar{u} z = 0$, $v \bar{z} + \bar{v} z = 0$ and thus $X_{H_1}(u_3) = X_{H_1}(v_3) = X_{H_2}(u_3) = X_{H_2}(v_3) = 0$.
Since $z=0$ also implies $u=v=0$, corresponding to rank-$0$ singularities, we have $z \ne 0$.

The equations $X_{H_1}(v) = x X_K(v)$ and $X_{H_2}(u) = y X_K(u)$ give after substituting $u,v$ and dividing by $z \ne 0$ the same equation
\[ -\frac{2g^2}{\delta_2-\delta_1}  (y u_3 - x v_3) = x y. \]

We then consider the equations $X_{H_1}(u) = x X_K(u)$ and $X_{H_2}(v) = y X_K(v)$. In the same way as above, we obtain 
\begin{align*}
& (\delta_1-\omega) x - 2 g^2 u_3 + \frac{2g^2}{\delta_2-\delta_1} (y u_3 - x v_3) = x^2, \\
& (\delta_2-\omega) y - 2 g^2 v_3 + \frac{2g^2}{\delta_2-\delta_1} (y u_3 - x v_3) = y^2, 
\end{align*}
which can be solved for $u_3$, $v_3$ to obtain
\begin{align*}
u_3 = - \frac{x}{2 g^2} (x + y - (\delta_1-\omega)), \ 
v_3 = - \frac{y}{2 g^2} (x + y - (\delta_2-\omega)).
\end{align*}
Since $-1 \le u_3 \le 1$ and $-1 \le v_3 \le 1$, these equations give constraints for $x,y$.

Moreover, we have
\begin{align*} 
z \bar{z} &= \frac{2g^2}{x^2} u \bar u = \frac{2g^2}{x^2} (1 - u_3^2)
=  \frac{2 g^2}{x^2} - \frac{(x+y-(\delta_1-\omega))^2}{2 g^2}, \\
z \bar{z} &= \frac{2g^2}{y^2} v \bar v = \frac{2g^2}{y^2} (1 - v_3^2)
=  \frac{2 g^2}{y^2} - \frac{(x+y-(\delta_2-\omega))^2}{2 g^2}.
\end{align*}
Note that there is an arbitrariness of a phase as a result of the fact that rank-$1$ singularities belong to orbits of the $\S^1$ action $\varphi_K$.

The equality of the two expressions for $z \bar{z}$ gives the following implicit relation between $x$ and $y$
\[ 4 g^4 (x^2-y^2) + (\delta_2 - \delta_1) x^2 y^2 (2(x + y) - (\delta_1 + \delta_2 - 2 \omega)) = 0. \]
Substituting the obtained expressions for $z \bar z$, $u_3$, $v_3$, $u$, $v$ into Eq.~\eqref{eq:K} for $K$ and Eq.~\eqref{eq:H1-H2-complex} for $H_1$, $H_2$ we then obtain the parameterization in Eq.~\eqref{eq:rank-1-param-general}. \qedhere
\end{proof}

\begin{remark}
The proof of Proposition~\ref{prop0:rank-1-parameterization} is inspired by the discussion of rank-$1$ singularities in \cite[Sec.~8.2]{Babelon2012} for the Tavis-Cummings system but does not make use of the spectral Lax pair formalism. In the spectral Lax pair formalism, the double roots of the spectral curve are not directly associated with physical quantities, and may lead to ``non-physical'' critical values. In \cite{Babelon2012} this issue is addressed by observing that for rank-$1$ singularities we have $X_{H_1} = x X_K$, $X_{H_2} = y X_K$ for some real numbers $x$, $y$ and showing that $x$, $y$ must satisfy certain constraints. This is also the starting point for our approach in the proof of Proposition~\ref{prop0:rank-1-parameterization} and it is sufficient for producing a full parameterization of rank-$1$ singularities without using the spectral Lax pair.
\end{remark}

Proposition~\ref{prop0:rank-1-parameterization} gives a parameterization of rank-$1$ singularities for the general TC system, except that we have an implicit relation between $x$ and $y$ in Eq.~\eqref{eq:relation-xy}. For the STC system, this implicit relation can be simplified and we can obtain the parameterization of rank-$1$ singularities given in the following Proposition. The rank-$1$ singularities of the STC system are shown in Figure~\ref{fig:stcm-bd}.

\begin{proposition}
\label{prop:rank-1-parameterization-stcm}
The rank-$1$ singularities of the special Tavis-Cummings system are parameterized as 
\begin{equation}\label{eq:rank-1-phase-space-special}
\begin{gathered}
z \bar z = \frac{2}{(a+b)^2} - \frac18 (1-4a)^2, 
\qquad
u = \frac{a-b}{\sqrt2 i} z,
\quad
v = \frac{a+b}{\sqrt2 i} z,
\\
u_3 = -\frac{1}{2} (a-b) \left( 2a + \frac12 \right), 
\quad 
v_3 = -\frac{1}{2} (a+b) \left( 2a - \frac12 \right),
\end{gathered}
\end{equation}
where we have the following possibilities for how $a$ depends on $b$, and the corresponding range of values of $b$:
\begin{enumerate}[label={(\roman*)}]
\item[$\ell_1$:] $a = 0$, $1/4 < b < 2^{2/3}$;
\item[$\ell_2$:] $a = 0$, $2^{2/3} < b < 4$;
\item[$\ell_3$:] $a = -(b^2 - 2 b^{1/2})^{1/2}$, $2^{2/3} < b < b^{\max}$; 
\item[$\ell_4$:] $a = (b^2 - 2 b^{1/2})^{1/2}$, $2^{2/3} < b < b^{\max}$; 
\item[$\ell_5$:] $a = 0$, $0 < b < 1/4$;
\item[$\ell_6$:] $a = 0$, $-4 < b < 0$;
\item[$\ell_7$:] $a = (b^2 + 2 b^{1/2})^{1/2}$, $0 < b < 1/4$;
\item[$\ell_8$:] $a = -(b^2 + 2 b^{1/2})^{1/2}$, $0 < b < 1/4$;
\item[$\mathfrak c_*$:] $a = 0$, $b = 2^{2/3}$;
\item[$\mathrm{hh}$:] $a = 0$, $b = 1/4$.
\end{enumerate}
Here, $b^{\max}$ is the unique real solution of the equation $16 b^3 - 8 b^2 + b - 64 = 0$, given in closed form by 
\[ b^{\max} = \frac{1}{12} \Bigl(2 + \bigl(3455 - 48 \sqrt{5181}\bigr)^{1/3} + \bigl(3455 + 48 \sqrt{5181}\bigr)^{1/3}\Bigr)  \approx 1.7583. \]
The corresponding critical values are parameterized in all cases as
\begin{subequations}\label{eq:rank-1-cv-special}
\begin{align}
k^c &=\frac12 (b - 4 a^2) + \frac{1}{(a-b)^2} - \frac{1}{16}(1+4a)^2, \\
h_1^c &=  2-\frac14(1+4a)(a-b)^2 + 2 \,\frac{1-2a}{a-b}, \\
h_2^c &= -2+\frac14(1-4a)(a+b)^2 + 2 \,\frac{1+2a}{a+b}. 
\end{align}
\end{subequations}
\end{proposition}

\begin{proof}
In the STC system we have $\delta_1 + \delta_2 = 2\omega$. Then Eq.~\eqref{eq:relation-xy} becomes
\[ 2 (x + y) [ 2 g^4 (x-y) + (\delta_2 - \delta_1) x^2 y^2] = 0. \]

Defining $x = a-b$, $y = a+b$, the last equation gives
\[ a ((\delta_2 - \delta_1)(a^2-b^2)^2-4g^4b) = 0, \]
with solutions $a = 0$, or $a^2=b^2 \pm 2 g^2 \sqrt{b/(\delta_2-\delta_1)}$. 
For the STC system we have $g = 1$ and $\delta_2-\delta_1=1$.
Therefore, the solutions become
\[ a = 0,\ b \in \R, \qquad a = \pm (b^2+2b^{1/2})^{1/2},\ b \ge 0, 
\qquad a = \pm (b^2-2b^{1/2})^{1/2},\ b \ge 2^{2/3}. \]
Substituting $x = a-b$, $y = a+b$ in Eq.~\eqref{eq:rank-1-phase-space-general} and Eq.~\eqref{eq:rank-1-param-general} we obtain Eq.~\eqref{eq:rank-1-phase-space-special} and Eq.~\eqref{eq:rank-1-cv-special}, respectively.

We consider now the constraints $-1 \le u_3 \le 1$, $-1 \le v_3 \le 1$. 
For the solution $a=0$ we obtain $u_3 = v_3 = b/4$. 
Therefore, the constraints give $-4 \le b \le 4$.
However, the values $b=\pm4$ correspond to the rank-$0$ singularities and are excluded, as is the value $b=0$ where $z \bar{z}$ is not defined. 
We denote by $\ell_6$ the case $a=0$, $b \in (-4,0)$.
We split the case $a=0$, $b \in (0,4)$ to several subcases: $\ell_5$ for $b \in (0,1/4)$, $\mathrm{hh}$ for $b=1/4$, $\ell_1$ for $b \in (1/4,2^{2/3})$, $\mathfrak c_*$ for $b = 2^{2/3}$, and $\ell_2$ for $b \in (2^{2/3},4)$.
The reason for this is that $\mathrm{hh}$ and $\mathfrak c_*$ correspond to degenerate singularities.

For the solutions $\ell_3$, $\ell_4$ with $a = \mp (b^2 - 2 b^{1/2})^{1/2}$ we find that they respectively give $u_3=-v_3=-1$ and $u_3=-v_3=1$ for $b = b^{\max} \approx 1.7583$. Therefore, we must have $2^{2/3} \le b \le b^{\max}$. The value $b = 2^{2/3}$ gives $a=0$ and corresponds to $\mathfrak c_*$ and is excluded from this parameterization. This also shows that $\ell_1$, $\ell_2$, $\ell_3$, $\ell_4$ meet at $\mathfrak c_*$ with $a=0$, $b=2^{2/3}$. The value $b = b^{\max}$ gives the rank-$0$ singularities $u_3=-v_3=-1$, $z=0$ for $\ell_3$ and $u_3=-v_3=1$, $z=0$ for $\ell_4$.

For the solutions $\ell_7$, $\ell_8$ with $a = \pm (b^2 + 2 b^{1/2})^{1/2}$ we find that they give $-1 \le u_3 \le 0$ and $-1 \le v_3 \le 0$ for $0 \le b \le 1/4$. The value $b=0$ is excluded because it gives $a=0$ and thus $z \bar z$ is not defined, while $b=1/4$ is excluded because it gives the rank-$0$ singularity $u_3=v_3=-1$, $z=0$.
\end{proof}

We now discuss the stability of the rank-$1$ singularities in the STC system. 

\begin{proposition}
\label{thm:rank-1-stability}
The rank-$1$ singularities of the STC system given in Proposition~\ref{prop:rank-1-parameterization-stcm} have the following types.
\begin{enumerate}[label={(\roman*)},noitemsep]
\item The families $\ell_1$, $\ell_2$, $\ell_3$, $\ell_4$ are FFR (focus-focus-regular).
\item The families $\ell_5$, $\ell_6$, $\ell_7$, $\ell_8$ are EER (elliptic-elliptic-regular).
\item The points $\mathrm{hh}$ and $\mathfrak c_*$ are degenerate.
\end{enumerate}
\end{proposition}

We note here that in the STC system when a family is EER the corresponding critical values lie at the boundary of the image of the integral map $F$, while FFR families form threads in the interior of the image of $F$.

\begin{proof}
To study the stability of the rank-$1$ singularities we perform a local reduction of the system, as described in Section~\ref{sec:reduction-s1-action}.
The corresponding reduced functions for the STC system are
\begin{align*}
\widehat H_1 &= - \frac{1}{2} u_3 - 2 u_3 v_3
+ 2 (1-u_3^2)^{1/2} \bigl[ (k-u_3-v_3)^{1/2} \cos\theta_u 
- (1-v_3^2)^{1/2} \cos(\theta_u-\theta_v) \bigr], \\
\widehat H_2 &= \frac{1}{2} v_3 + 2 u_3 v_3
+ 2 (1-v_3^2)^{1/2} \bigl[ (k-u_3-v_3)^{1/2} \cos \theta_v 
+ (1-u_3^2)^{1/2} \cos(\theta_u-\theta_v) \bigr].
\end{align*}

The stability of the rank-$1$ singularities is determined by analyzing the eigenvalues of the derivatives $DX_{\widehat H_1}$, $DX_{\widehat H_2}$ evaluated at the singularity.
For reasons of computational convenience we consider instead of $X_{\widehat H_1}$ and $X_{\widehat H_2}$, their linear combinations
\[ X_{\widehat H_+} = X_{\widehat H_1} + X_{\widehat H_2}, \quad X_{\widehat H_-} = X_{\widehat H_1} - X_{\widehat H_2}, \]
and their derivatives $DX_{\widehat H_+}$, $DX_{\widehat H_-}$.
We report the results for each family of rank-$1$ singularities separately. 

\textbf{Families $\ell_1$, $\ell_2$, $\ell_5$, $\ell_6$.}
For these families, the characteristic polynomials for $DX_{\widehat H_+}$ and $DX_{\widehat H_-}$ are respectively
\[ P_1(r) = \frac{1}{b^4}(b^2 r^2 - (b^3-4))^2, \quad P_2(r) = \frac{1}{b^4}(b^2 r^2 - (b^3-4)(1-4b))^2. \]
The form of the characteristic polynomials implies that each matrix has at most two distinct eigenvalues with algebraic multiplicity $2$ and they are either both real $(\pm \alpha)$ or both imaginary $(\pm i \beta)$. 
If $-4 < b < 0$ ($\ell_6$), or if $0 < b < 1/4$ ($\ell_5$), then both $DX_{\widehat H_+}$ and $DX_{\widehat H_-}$ have imaginary eigenvalues and the singularity has type EER.
If $1/4 < b < 2^{2/3}$ ($\ell_1$) then $DX_{\widehat H_+}$ has imaginary eigenvalues and $DX_{\widehat H_-}$ has real eigenvalues. 
In this case, the singularity has type FFR.
If $2^{2/3} < b < 4$ ($\ell_2$), then $DX_{\widehat H_+}$ has real eigenvalues and $DX_{\widehat H_-}$ has imaginary eigenvalues. 
In this case, the singularity has again type FFR.

If $b=1/4$, then $DX_{\widehat H_+}$ has imaginary eigenvalues, while $DX_{\widehat H_-}$ has four $0$ eigenvalues; this degenerate point marks the transition from the EER family $\ell_5$ to the FFR family $\ell_1$.
Finally, if $b=2^{2/3}$, then both matrices have four $0$ eigenvalues; this is the most degenerate case corresponding to $\mathfrak c_*$.

\textbf{Families $\ell_3$, $\ell_4$.}
For these families, the characteristic polynomials for $DX_{\widehat H_+}$ and $DX_{\widehat H_-}$ are respectively
\begin{align*}
P_1(r) &= r^4 
- \frac{2}{b^{1/2}} (-b^{3/2}+2) (6b+1) r^2
+ \frac{2}{b} (-b^{3/2}+2) (1 + 12 b - 4 b^{3/2} (6b-1)), \\
P_2(r) &= r^4 
- \frac{2}{b^{1/2}} (-b^{3/2}+2) (32b^2 - 10b - 16b^{1/2} + 1) r^2 \\
& \qquad \qquad + \frac{2}{b} (-b^{3/2}+2) (4b-1)^2 (1 + 12 b - 4 b^{3/2} (6b-1)).
\end{align*}
For all $b \in (2^{2/3},b^{\max})$ the eigenvalues of both $DX_{\widehat H_+}$ and $DX_{\widehat H_-}$ are complex with non-zero real and imaginary parts, and the type of the singularity is FFR.

\textbf{Families $\ell_7$, $\ell_8$.}
For these families, the characteristic polynomials for $DX_{\widehat H_+}$ and $DX_{\widehat H_-}$ are respectively
\begin{align*}
P_1(r) &= r^4 
+ \frac{2}{b^{1/2}} (b^{3/2}+2) (6b+1) r^2
+ \frac{2}{b} (b^{3/2}+2) (1 + 12 b + 4 b^{3/2} (6b-1)), \\
P_2(r) &= r^4 
+ \frac{2}{b^{1/2}} (b^{3/2}+2) (32b^2 - 10b + 16b^{1/2} + 1) r^2 \\
& \qquad \qquad + \frac{2}{b} (b^{3/2}+2) (4b-1)^2 (1 + 12 b + 4 b^{3/2} (6b-1)).
\end{align*}
For all $b \in (0,1/4)$ the eigenvalues of both $DX_{\widehat H_+}$ and $DX_{\widehat H_-}$ are imaginary and the type of the singularity is EER.
\end{proof}

\subsection{Rank-2 singularities}
\label{sec:rank-2-singularities}

The rank-$2$ critical values of the general TC system are given in the following result whose proof is presented in Appendix~\ref{App.Rank2}. The rank-$2$ singularities for the STC system are depicted in Figure~\ref{fig:stcm-bd}.

\begin{proposition}\label{prop:rank-2-singularities}
The rank-$2$ critical values of the general Tavis-Cummings system on $K^{-1}(k)$ are parameterized by
\begin{subequations}
\begin{align}
\label{eq:rank-2:H1}
H_1 &= \frac{1}{\delta_2 - \delta_1} \left[ -\frac{k}{x^2} - \frac{C^2 g^2}{x^2} + \frac{C}{x^2 y} + g^2 \Bigl( \frac{y^2}{x^2} + 2 \frac{x}{y} - 1 \Bigr) \right],
\\
\label{eq:rank-2:H2}
H_2 &= -\frac{1}{\delta_2 - \delta_1} \left[-\frac{k}{y^2} - \frac{C^2 g^2}{y^2} + \frac{C}{x y^2} + g^2 \Bigl( \frac{x^2}{y^2} + 2 \frac{y}{x} - 1 \Bigr) \right], 
\end{align}
\end{subequations}
where $x,y \in \R$ satisfy $x - y = (\delta_2 - \delta_1) x y$, and
\begin{align*}
C = \frac{1}{2g^2} (\delta_1-\omega) - \frac{1}{2g^2x}
= \frac{1}{2g^2} (\delta_2-\omega) - \frac{1}{2g^2y}.
\end{align*}
The parameter $x$ takes only those values for which the cubic polynomial $P(x,y;u_3) = a_3(x,y) u_3^3 + a_2(x,y)u_3^2 + a_1(x,y) u_3 + a_0(x,y)$ has two roots in $[-1,1]$. Here,
\begin{align*}
a_3(x,y) & = -4 g^2 x^2 y (x-y) \\
a_2(x,y) & = -4 C^2 g^4 x^2 y^2+8 C g^2 x^2 y-4 C g^2 x y^2-4 g^2 K x^2 y^2
-(x-y)^2 \\
a_1(x,y) & = 2 (2 C^3 g^4 x y^2-3 C^2 g^2 x y+C^2 g^2 y^2+2 C g^4 x^3
   y^2-2 C g^4 x y^4+2 C g^2 K x y^2 \\
& \qquad +C x-C y+g^2 x^3 y-g^2 x^2 y^2+g^2 x y^3-g^2 y^4-K x y+K y^2 ) \\
a_0(x,y) & =  -C^4 g^4 y^2+2 C^3 g^2 y-2 C^2 g^4 x^2 y^2+2 C^2 g^4 y^4-2 C^2 g^2 K y^2 - C^2 -2 C g^2 x^2 y \\
& \qquad -2 C g^2 y^3+2 C K y-g^4 x^4 y^2+2 g^4 x^2 y^4-g^4 y^6+2 g^2 K x^2 y^2+2 g^2 K y^4 -K^2 y^2.
\end{align*}
\end{proposition}

\section{Hamiltonian monodromy}
\label{sec:global-topology}

We consider the Hamiltonian monodromy of the STC system.
First, we discuss the fundamental definitions and properties of Hamiltonian monodromy for 3-DOF integrable Hamiltonian systems. 
Then, we give the Hamiltonian monodromy matrices for the special Tavis-Cummings system based on numerical calculations.
Finally, we discuss the implications of these results for the reduced 2-DOF Hamiltonian systems on $K^{-1}(k)/\S^1$.

\subsection{Hamiltonian monodromy in 3-DOF integrable Hamiltonian systems}
\label{sec:monodromy-3dof}

Consider a 3-DOF integrable Hamiltonian system $F : M \to \R^3$, denote by $R$ the set of regular values of $F$, and assume for simplicity that each regular fiber $F^{-1}(r)$, $r \in R$, is connected and compact and thus $F^{-1}(r) \simeq \T^3$. 
Recall that, for $r_0 \in R$, the \emph{monodromy map}
\[ \mu : \pi_1(R,r_0) \to \group{Aut}(F^{-1}(r_0)), \] 
maps each homotopy class $\gamma$ with base point $r_0$, to an element $\mu(\gamma)$ in the group $\group{Aut}(F^{-1}(r_0))$ of orientation preserving automorphisms of the fiber $F^{-1}(r_0)$.  
The latter group is isomorphic to $\SL(3,\Z)$.

The element $\mu(\gamma) \in \SL(3,\Z)$ is defined via parallel transport of the period lattice of $r_0$ along a representative of the homotopy class $\gamma$.
We recall here the construction.
Given a connected regular fiber $F^{-1}(r) \simeq \T^3$, the \emph{period lattice} $\Lambda_r$ is defined as
\[ \Lambda_r = \bigl\{ \vv T = (T_1,T_2,T_3) \in \R^3 
: \varphi_{H_1}^{T_1} \circ \varphi_{H_2}^{T_2} \circ \varphi_K^{T_3} = \mathrm{id}_{F^{-1}(r)} \bigr\}.  \]
The period lattice $\Lambda_r$ is isomorphic to $\Z^3$. 
The \emph{period lattice bundle} is the disjoint union
\[ \Lambda = \bigsqcup_{r \in R} \Lambda_r. \]
Over each simply connected open subset $U \subseteq R$, the restriction of the period lattice bundle $\Lambda|_U$ is a trivial bundle, that is,
\[ \Lambda|_U = \bigsqcup_{r \in U} \Lambda_r \simeq U \times \Z^3. \]
Therefore, there is a well defined notion of \emph{parallel transport} of elements of the period lattice bundle along paths in $R$.

\begin{remark}\label{rem:local-action}
If $\vv T = (T_1,T_2,T_3) \in \Lambda_r$, then the vector field $X_{\vv T} = T_1 X_{H_1} + T_2 X_{H_2} + T_3 X_K$
generates a period-$1$ flow $\varphi$ on $F^{-1}(r)$ defined by $\varphi^s = \varphi_{H_1}^{sT_1} \circ \varphi_{H_2}^{sT_2} \circ \varphi_K^{sT_3}$.
Moreover, given a smooth section $\vv T: U \to \Lambda|_U : r \mapsto \vv T(r)$, the vector field $X_{\vv T(r)}$ is Hamiltonian with generating function $2 \pi I$, where $I$ is a (local) action coordinate for the system.
\end{remark}

Given a loop $\gamma$, based at $r_0 \in R$, the parallel transport of a period vector $\vv T \in \Lambda_{r_0}$ along $\gamma$ results to a possibly different period vector $\vv T' \in  \Lambda_{r_0}$, and thus induces a $\Z$ linear map $\mu(\gamma) : \Lambda_{r_0} \to \Lambda_{r_0}$.

Given a basis $\vv T_1, \vv T_2, \vv T_3$ of $\Lambda_{r_0}$, the map $\mu(\gamma)$ is fully described by its action on the given basis. In particular, it can be described through a matrix $\widetilde M(\gamma) \in \SL(3,\Z)$ such that
\[ \vv T' = \mu(\gamma)(\vv T) = \widetilde M(\gamma) \vv T. \]
Suppose that $\vv T_i' = \sum_j m_{ij} \vv T_j$, where $m_{ij}$ are the elements of a matrix $M(\gamma)$.

\begin{lemma}
$M(\gamma) = \widetilde M(\gamma)^t$.
\end{lemma}

\begin{proof}
Write $\vv T = \sum_j a_j \vv T_j$ and $\vv T' = \sum_i b_i \vv T_i$. 
The linearity of $\mu(\gamma)$ implies $\vv T' = \sum_j a_j \vv T_j'$,
and we find $\sum_j \sum_i a_j m_{ji} \vv T_i = \sum_i b_i \vv T_i$, 
implying $b_i = \sum_j m_{ji} a_j$. 
However, by definition, $b_i = \sum_j [M(\gamma)]_{ij} a_j$, implying that the matrix $M(\gamma)$ relating the two bases is the transpose of the matrix $\widetilde M(\gamma)$ expressing the automorphism $\mu(\gamma)$.
\end{proof}

Given two loops (or homotopy classes) $\gamma$, $\gamma'$ based at $r_0 \in R$ we denote by $\gamma \cdot \gamma'$ the loop (or homotopy class) defined by first traversing $\gamma$ and then $\gamma'$. 
Then the map $\gamma \mapsto M(\gamma)$ induces a homomorphism from $\pi_1(R,r_0)$ to $\SL(3,\Z)$ with the property that 
\[ M(\gamma \cdot \gamma') = M(\gamma) M(\gamma'). \]

\begin{lemma}
Suppose that $\{\vv T_i\}_{i=1,2,3}$ and $\{\vv S_i\}_{i=1,2,3}$ are bases of $\Lambda_{r_0}$ satisfying $\vv S_i = \sum_j A_{ij} \vv T_j$ with $A = (A_{ij}) \in \SL(3,\Z)$, and denote by $M_{\vv T}(\gamma)$ and $M_{\vv S}(\gamma)$ the monodromy matrices representing $\mu(\gamma)$ in each of the corresponding bases. Then
\[ M_{\vv S}(\gamma) = A M_{\vv T}(\gamma) A^{-1}. \]
\end{lemma}

\begin{proof}
Write $M_{\vv T}(\gamma) = (M_{ij})$ and $M_{\vv S}(\gamma) = (N_{ij})$.
Then we have
\[ \mu(\gamma)(\vv S_i) = \sum_j A_{ij} \mu(\gamma)(\vv T_j)
= \sum_j \sum_k A_{ij} M_{jk} \vv T_k 
= \sum_l \Bigl( \sum_j \sum_k  A_{ij} M_{jk}  (A^{-1})_{kl} \Bigr) \vv S_l. \]
Since $\mu(\gamma)(\vv s_i) = \sum_j N_{il} \vv S_l$, we obtain
\[ N_{il} = \sum_j \sum_k  A_{ij} M_{jk} (A^{-1})_{kl}, \]
that is, $M_{\vv S}(\gamma) = A M_{\vv T}(\gamma) A^{-1}$.
\end{proof}

\subsection{Hamiltonian monodromy in the special Tavis-Cummings system}
\label{sec:monodromy-tavis-cummings}

We now specialize the discussion above to the special Tavis-Cummings system.

\begin{proposition}
\label{prop:fund-group-tc}
For each $r_0 \in R$, where $R$ is the set of regular values of the special Tavis-Cummings system, the fundamental group $\pi_1(R,r_0)$ is isomorphic to  the free group of three generators $\set{F}_3 = \Z * \Z * \Z$.
\end{proposition}

\begin{proof}
The set $R$ is homeomorphic to $\R^3 \setminus (\{x=0\} \cup \{y=0\})$ and the latter deformation retracts to the unit sphere with four punctures, denoted by $S$. In turn, $S$ deformation retracts to the open unit disk with three punctures, and that to the wedge sum of three circles, the fundamental group of which is isomorphic to $\set{F}_3$ \cite[Example~1.21]{Hatcher2002}.
\end{proof}

Since $R$ is path-connected we can write $\pi_1(R)$ instead of $\pi_1(R,r_0)$. 
We construct representatives of four non-trivial homotopy classes of $\pi_1(R)$ based at $r_0$.
First, choose a value $K = k_1 = 1.8$ and consider the regular value $r_0 = (2, 1, k_1)$ of the integral map $F$. 
Since $K_* < k_1 < 2$, the plane $K = k_1$ contains two focus-focus-regular singularities $\ff_1 \approx (-0.578466, 0.578466, k_1)$ and $\ff_2 \approx (-1.74301, 1.74301, k_1)$, see Figure~\ref{fig:fundamental-group} (left).
Then, choose a value $K = k_2 = 0.3$. 
Since $0 < k_2 < K_*$, the plane $K = k_2$ contains the focus-focus-regular singularities $\ff_3 \approx (1.79368, -1.16404, k_2)$ and $\ff_4 \approx (1.16404, -1.79368, k_2)$, see Figure~\ref{fig:fundamental-group} (center).
Then, define loops $\gamma_j$, $j=1,2,3,4$ in the following way. Each $\gamma_j$ consists of a straight line joining $r_0$ to the point $\ff_j + (L,0,0)$, followed by a circular path of radius $L$ centered at $\ff_j$ and parameterized by $\ff_j + (L\cos(2\pi s), L\sin(2\pi s), 0)$, $0 \le s \le 1$, and finally followed by a straight line back to $r_0$. Notice that $\gamma_1$, $\gamma_2$ lie on the plane $K = k_1$.
The loops $\gamma_j$, $j=1,2,3,4$ are shown in Figure~\ref{fig:fundamental-group}. 

The corresponding homotopy classes satisfy the relation
\[ \gamma_2 \cdot \gamma_1 = \gamma_3 \cdot \gamma_4. \]
Therefore, any three of these four homotopy classes can be chosen as generators of $\pi_1(R)$.
Moreover, notice that $\gamma_1 \cdot \gamma_2 \ne \gamma_2 \cdot \gamma_1$, $\gamma_3 \cdot \gamma_4 \ne \gamma_4 \cdot \gamma_3$, and $\gamma_2 \cdot \gamma_1 \ne \gamma_4 \cdot \gamma_3$.

\begin{figure}
\centering
\includegraphics[width=5cm]{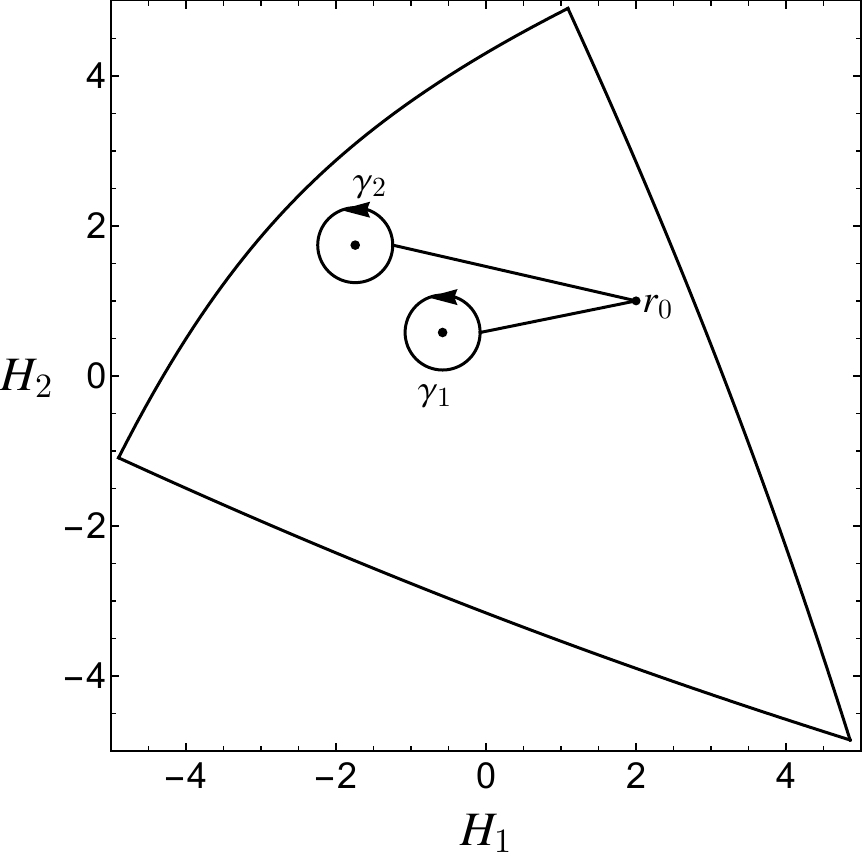}\hfill
\includegraphics[width=5cm]{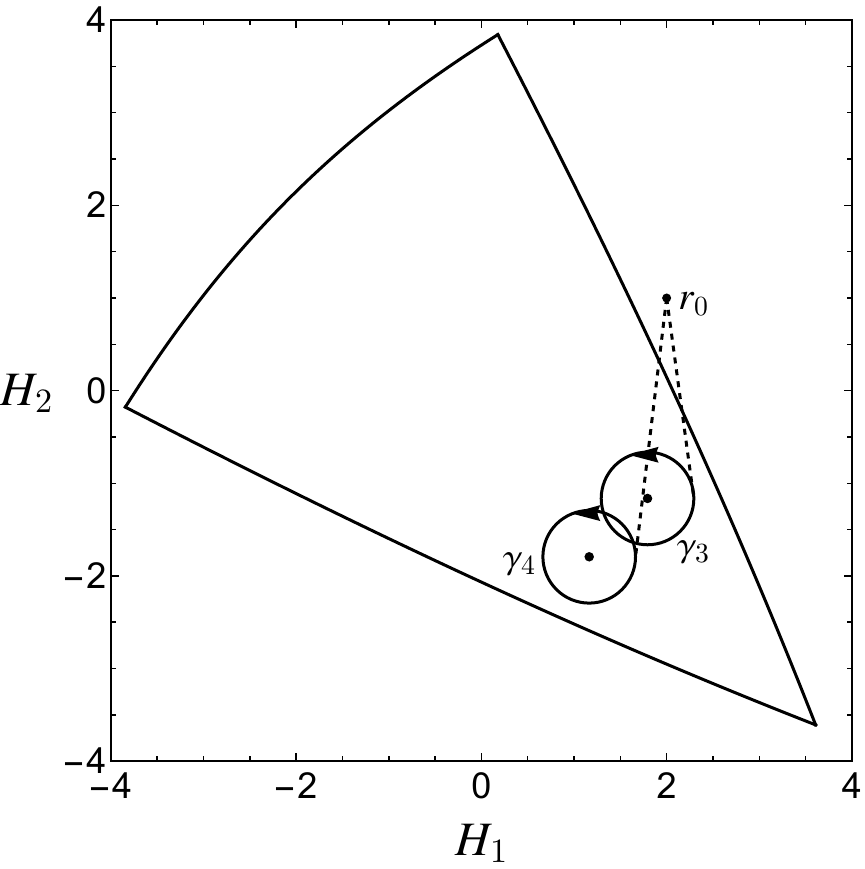}\hfill
\includegraphics[width=5cm]{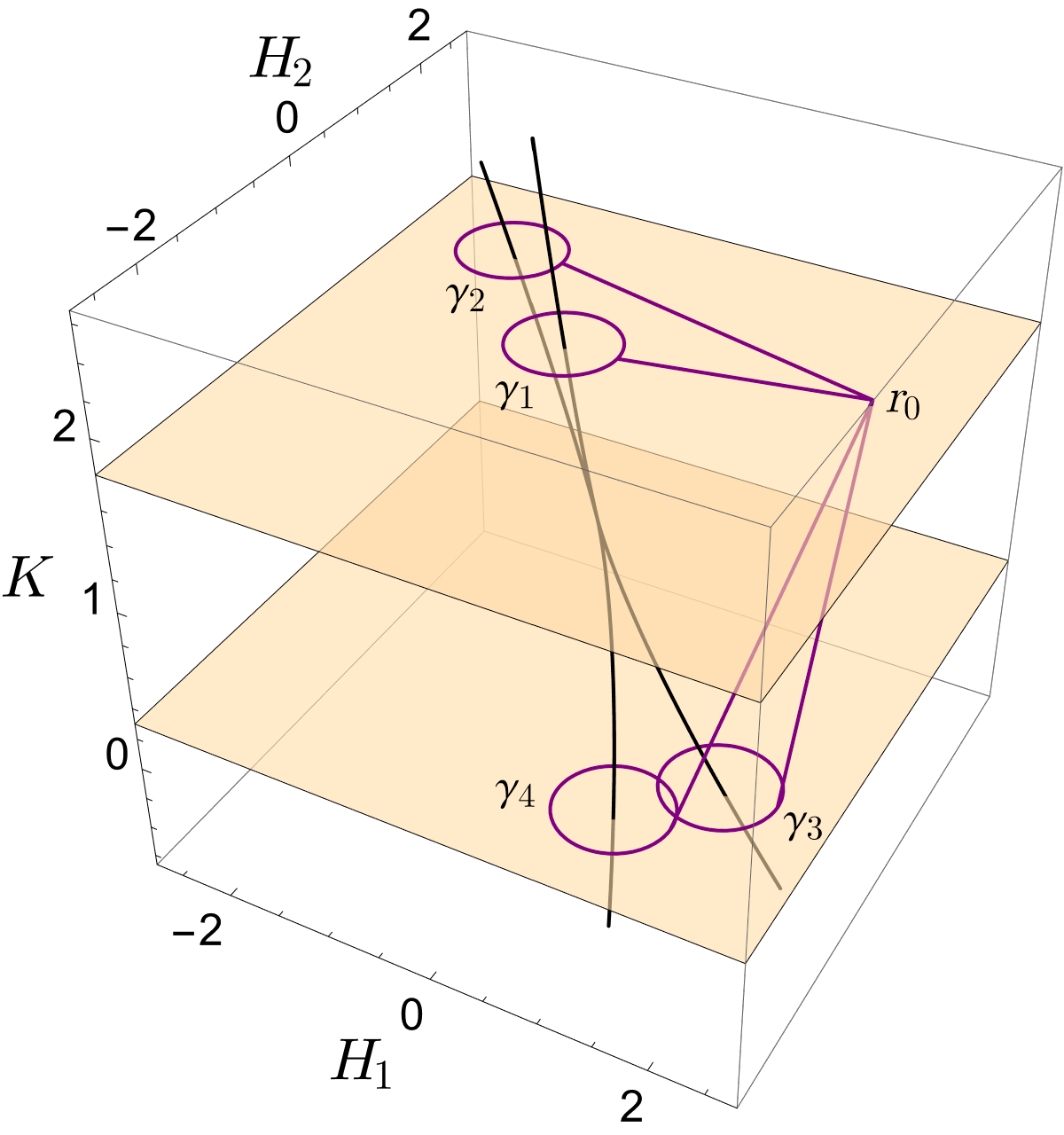}
\caption{Homotopically non-trivial loops $\gamma_1$, $\gamma_2$, $\gamma_3$, $\gamma_4$ in the set of regular values, based at the point $r_0 = (2,1,1.8)$. The left and center panels show the planes $K=1.8$ and $K=1.3$ respectively. In the center panel, the point $r_0=(2,1,1.8)$ and the dashed straight lines lie outside the plane $K=0.3$. The right panel shows the loops in the three-dimensional space $(H_1,H_2,K)$.
}
\label{fig:fundamental-group}
\end{figure}

To compute the monodromy matrices $M(\gamma_j)$, $j=1,2,3,4$ for the four loops defined above we use the numerical procedure detailed in Appendix~\ref{sec:numerical-computation}.
In short, we numerically compute a basis $\vv T_1 = (0, 0, 2\pi)$, $\vv T_2$, $\vv T_3$ of the period lattice $\Lambda_{r_0}$ on the fiber $F^{-1}(r_0)$ and then we numerically continue this period lattice basis along a loop $\gamma$ back to $r_0$.
The numerical continuation along $\gamma$ produces a new period lattice basis $\vv T_1' = (0, 0, 2\pi)$, $\vv T_2'$, $\vv T_3'$, and the corresponding Hamiltonian monodromy matrix is then given by
\[ M = 
\mathcal{B}(\vv{T}_1', \vv{T}_2', \vv{T}_3') \,
\mathcal{B}(\vv{T}_1, \vv{T}_2, \vv{T}_3)^{-1}, \]
where $\mathcal{B}(\vv T_1, \vv T_2, \vv T_3)$ is the $3 \times 3$ matrix with rows given by the (row) vectors $\vv T_1$, $\vv T_2$, and $\vv T_3$.

Using this approach we computed the monodromy matrices for the loops $\gamma_j$, $j=1,\dots,4$ as
\begin{gather*}
M(\gamma_1) = \begin{bmatrix}
1 & 0 & 0 \\
0 & 1 & 1 \\
0 & 0 & 1 
\end{bmatrix}, \ 
M(\gamma_2) = M(\gamma_4) =
\begin{bmatrix}
 1 & 0 & 0 \\
 0 & 1 & 0 \\
 0 & -1 & 1 
\end{bmatrix}, \ 
M(\gamma_3) =
\begin{bmatrix}
 1 & 0 & 0 \\
 0 & 2 & 1 \\
 0 & -1 & 0
\end{bmatrix}.
\end{gather*}

Notice that
\begin{align*}
M(\gamma_2)M(\gamma_1) = M(\gamma_3)M(\gamma_4)
= \begin{bmatrix}
1 & 0 & 0 \\
0 & 1 & 1 \\
0 & -1 & 0 \\
\end{bmatrix},
\end{align*}
consistent with the relation $\gamma_2 \cdot \gamma_1 = \gamma_3 \cdot \gamma_4$.

\begin{remark}
This result implies that the Lagrangian fibration induced by $F$ over its regular values is not globally trivializable~\cite{Duistermaat1980}. Equivalently, the system does not admit globally defined action–angle coordinates. Moreover, the fact that there is no change of basis which can simultaneously bring the monodromy matrices into the form 
\[
\begin{bmatrix}
1 & 0 & 0 \\
0 & 1 & 0 \\
a & b & 1 \\
\end{bmatrix}
\]
confirms that the system does not have a global $\T^2$ action.
It is known, however, that near a thread of focus-focus-regular values there is a local $\T^2$ action. 
We can read from the form of $M(\gamma_1)$ that the vector fields $X_{\vv T_1} = 2\pi X_J$ and $X_{\vv T_3}$ extend smoothly inside the disk bounded by $\gamma_1$ and generate a local $\T^2$ action in the neighborhood of $\ell_1$, see Remark~\ref{rem:local-action}.
Similarly, we read from $M(\gamma_2) = M(\gamma_4)$ that the vector fields $X_{\vv T_1}$ and $X_{\vv T_2}$ generate a local $\T^2$ action in the neighborhood of $\ell_2$ and $\ell_4$.
Finally, we read from $M(\gamma_3)$ that the vector fields $X_{\vv T_1}$ and $X_{\vv T_2 + \vv T_3}$ generate a local $\T^2$ action in the neighborhood of $\ell_3$.
\end{remark}

\subsection{Monodromy of \texorpdfstring{$\S^1$}{S1} reduced 2-DOF Hamiltonian systems}
\label{sec:monodromy-reduced}

Recall that for $k \not\in \{-2,0,2\}$ the reduced space $K^{-1}(k)$ is a smooth manifold, see Section~\ref{sec:reduction-s1-action} and Proposition~\ref{prop:reduced-spaces}. The corresponding reduced functions $\widehat H_1$, $\widehat H_2$ define a reduced, 2-DOF integrable Hamiltonian system on $K^{-1}(k)/\S^1$. We give here some brief observations about these reduced systems. Note that representative bifurcation diagrams for the cases discussed here are shown in Figure~\ref{fig:stcm-bd} as horizontal slices of the three-dimensional bifurcation diagram with the planes $K = k$.

First, note that for $-2 < k < 0$ and for $k > 16.0625$, the corresponding reduced system has no focus-focus points and is of toric type.
For $2 < k < 16.0625$, the reduced system has a single focus-focus point and, therefore, it has a global $\S^1$ action on the manifold $K^{-1}(k) / \S^1 \simeq \S^2 \times \S^2$. 

We consider now the situation for $0 < k < 2$.
For $0 < k < k^*$, where $k^* \approx 1.12805$ corresponds to the central singularity, the reduced system has two focus-focus points and no global $\S^1$ action. The monodromy matrices around the two focus-focus points can be obtained by eliminating the first row and first column from the monodromy matrices computed in Section~\ref{sec:monodromy-tavis-cummings}. We obtain
\begin{gather*}
\widehat M(\gamma_3) =
\begin{bmatrix}
2 & 1 \\ -1 & 0
\end{bmatrix}, \ 
\widehat M(\gamma_4) =
\begin{bmatrix}
1 & 0 \\ -1 & 1 
\end{bmatrix}.
\end{gather*}

For $k^* < k < 2$, the reduced system has again two focus-focus points and no global $\S^1$ action. The monodromy matrices are
\begin{gather*}
\widehat M(\gamma_1) = \begin{bmatrix}
1 & 1 \\
0 & 1 
\end{bmatrix}, \ 
\widehat M(\gamma_2) = 
\begin{bmatrix}
1 & 0 \\
-1 & 1 
\end{bmatrix}.
\end{gather*}

At the transition point $k = k^*$ between the previous two cases, the reduced system has a single degenerate critical value in the interior of the image of the reduced integral map. This critical value corresponds to the $A_2$ singularity and the corresponding reduced fiber is homeomorphic to $\S^2$, see Proposition~\ref{prop:a2-sphere}.
As in the previous two cases, the reduced system her does not admit a global $\S^1$ action. The monodromy matrix for a loop $\gamma_0$ around this critical value, obtained by homotopy from $\gamma_2 \cdot \gamma_1$ or $\gamma_3 \cdot \gamma_4$, is
\begin{align*}
\widehat M(\gamma_0) = \begin{bmatrix}
1 & 1 \\
-1 & 0 \\
\end{bmatrix},
\end{align*}
for which we note that its eigenvalues are $-e^{\pm 2\pi i / 3}$ and therefore it cannot be brought into the standard form $\bigl(\begin{smallmatrix} 1 & 0 \\ m & 1 \end{smallmatrix}\bigr)$ by a change of basis.

\section{\texorpdfstring{$\boldsymbol{A_2}$}{A2} singularity}
\label{sec:a2-singularity}

In this section, we focus on the critical value $\mathfrak c_*$ at the meeting point of the four threads $\ell_j$, $j=1,\dots,4$ of focus-focus-regular critical values, see Proposition~\ref{prop:rank-1-parameterization-stcm}. The corresponding fiber $F^{-1}(\mathfrak c_*)$ contains an $\S^1$ orbit of degenerate rank-$1$ critical points. Upon reducing the $\S^1$ action, this $\S^1$ orbit corresponds to a singularity $c_*$ in the reduced space. Following Thom's terminology~\cite{thom2018structural}, the degenerate singularity $ c_*$ acts as the \emph{organizing center} of the bifurcations. We thus refer to $\mathfrak c_*$  as the \emph{central critical value}, and to the corresponding critical point in the reduced space, $c_*$, as the \emph{central singularity}.

\subsection{Reduction to normal form}

Recall from Section~\ref{sec:reduction-s1-action} that the $\S^1$ action induced by the momentum $K$ on the subspace $M_* = \S^2_{*,u} \times \S^2_{*,v} \times \R^2_*$ can be reduced by introducing coordinates $(\theta_u,u_3,\theta_v,v_3,k)$.
Here, we do not restrict attention to $K^{-1}(k)$, that is, we do not fix $k$.
Then, the integral map $F = (H_1,H_2,K)$ reduces to the map $\widehat F = (\widehat H_1, \widehat H_2, \widehat K)$, where $\widehat H_1$, $\widehat H_2$ are given by Eq.~\eqref{eq:H1-H2-reduced} and $\widehat K(\theta_u,u_3,\theta_v,v_3,k) = k$.

It follows from the parameterization in Proposition~\ref{prop:rank-1-parameterization-stcm} that the central singularity $c_*$ corresponds to the values $\theta_u^* = \pi$, $\theta_v^* = 0$, $u_3^* = v_3^* = 1 / (2 \cdot 2^{1/3})$, and $k^* = (12 \cdot 2^{2/3}-1)/16$, and we recall that the corresponding critical value is
\[ \mathfrak c_* = \widehat F(c_*) = \Bigl(
2 - 3 \cdot 2^{-2/3}, -2 + 3 \cdot 2^{-2/3},
\tfrac{1}{16} (12 \cdot 2^{2/3} - 1) \Bigr). \]

Let $\mathcal F : \S^2_u \times \S^2_v \times \R^2$ be the integral map given by
\[ \mathcal F = \psi \circ F, \]
where $\psi : \R^3 \to \R^3$ is the diffeomorphism
\[ \psi(h_1,h_2,k) = \Bigl(h_1+h_2,\, 
\frac{h_2-h_1+2\eta(k-k^*)}{(4\cdot2^{2/3}-1)^{1/2}}, \, \frac{2^{14/9} (k-k^*)}{(4\cdot2^{2/3}-1)^{1/3}} \Bigr), \]
and
\[ \eta(k-k^*) = 2-3\cdot2^{-2/3} - 2^{2/3} (k-k^*)
+ \frac{4\cdot2^{2/3}}{4\cdot2^{2/3}-1} (k-k^*)^2
- \frac{32 \cdot 2^{2/3}}{(4\cdot2^{2/3}-1)^2} (k-k^*)^3. \]
The function $\eta(k-k^*)$ is the cubic Taylor polynomial in $k-k^*$ of the expression obtained by substituting $\theta_u=\theta_u^*$, $u_3=u_3^*$, $\theta_v=\theta_v^*$, $v_3 = v_3^*$ in $\widehat H_1$.
It can be directly checked that $\psi(\mathfrak c_*) = 0$.

Since $\psi$ is a diffeomorphism, the singular Lagrangian fibrations induced by $F$ and $\mathcal F$ are isomorphic.
In particular, $F^{-1}(v) = \mathcal F^{-1}(\psi(v))$.
The integral map $\mathcal F$ reduces under the $\S^1$ action to $\widehat{\mathcal F} = \psi \circ \widehat F$.
In the rest of this section we work with the reduced integral map $\widehat{\mathcal F}$ in a neighborhood of its critical value $0$.

\begin{theorem}\label{Th:local}
There exists a local fibered bundle isomorphism $(\Phi,\phi)$ between the following local fibrations:
\[
\begin{tikzcd}[row sep=3em, column sep=4em]
\bigl(\S^2_{*,u} \times \S^2_{*,v} \times [-2,\infty), c_*\bigr) \arrow[d, "\widehat{\mathcal F}"'] \arrow[r,"\Phi"]  &  (\C^2\times\R,0) \arrow[d, "f"] \\
\bigl(\R^3, 0 \bigr) \arrow[r,"\phi"] &  (\C \times \R, 0) ,
\end{tikzcd}
\]
where $f(x,y,\kappa)=(y^2+x^3+\kappa x,\kappa)$.
\end{theorem}

\begin{proof}
Write $\widehat{\mathcal F} = (\widehat H_+, \widehat H_-, \widehat L)$, replace $\theta_u = \theta_u^* + \xi_u$, $\theta_v = \theta_v^* + \xi_v$, $u_3 = u_3^* + \chi_u$, $v_3 = v_3^* + \chi_v$, and consider the Taylor expansions of $\widehat H_+$ and $\widehat H_-$ with respect to $\xi_u$, $\xi_v$, $\chi_u$, $\chi_v$, $k-k^*$.

Up to linear terms the Taylor expansions are $\widehat H_+^{(1)} = \widehat H_-^{(1)} = 0$.
The quadratic terms are given by
\begin{align*}
\widehat H_+^{(2)} &=  \frac{4\cdot2^{2/3}-1}{8 \cdot 2^{1/3}} (\xi_u^2 - \xi_v^2)
+ \frac{4}{4 \cdot 2^{2/3} - 1} (k-k^*) (\chi_u - \chi_v), \\
\widehat H_-^{(2)} &= \frac{(4\cdot2^{2/3}-1)^{1/2}}{8 \cdot 2^{2/3}} \bigl[ (4-2^{1/3}) (\xi_u^2 + \xi_v^2) - 8 \xi_u \xi_v \bigr] 
+ \frac{4}{(4\cdot2^{2/3}-1)^{1/2}} (k-k^*) (\chi_u + \chi_v).
\end{align*}
We observe that the quadratic terms for $k=k^*$ do not involve $\chi_u$, $\chi_v$. 
This implies, in particular, that the central singularity is degenerate.

We then compute the cubic terms $\widehat H_+^{(3)}$, $\widehat H_-^{(3)}$ in the Taylor expansion of $\widehat H_+$, $\widehat H_-$, and we keep only these terms corresponding to the Newton polytope for each of these two functions.
That is, we consider
\[
\widetilde H_+ =  \widehat H_+^{(2)} + \widetilde H_+^{(3)},
\quad 
\widetilde H_- = \widehat H_-^{(2)} + \widetilde H_-^{(3)}, 
\]
where
\begin{align*}
\widetilde H_+^{(3)}
& = 
\frac{8}{(4\cdot2^{2/3}-1)^2}\bigl[
(6-2^{1/3})(\chi_v^3-\chi_u^3) - 6 \chi_u \chi_v(\chi_u-\chi_v)
\bigr],
\\
\widetilde H_-^{(3)} 
& = 
-\frac{8}{(4\cdot2^{2/3}-1)^{3/2}}\bigl[ 
(2-2^{1/3})(\chi_u^3+\chi_v^3) + 6 (\chi_u+\chi_v) \chi_u \chi_v
\bigr].
\end{align*}
The fibrations induced by $\widehat{\mathcal F} = (\widehat H_+, \widehat H_-, \widehat L)$ and $\widetilde{\mathcal F} = (\widetilde H_+, \widetilde H_-, \widehat L)$ are locally topologically equivalent in a neighborhood of the point $(\theta_u^*,u_3^*,\theta_v^*,v_3^*,k^*)$.

Consider complex coordinates $(x,y) \in \C^2$, write $x = x_1 + i x_2$, $y = y_1 + i y_2$, and define the function $h : \C^2 \times \R \to \C$ by
\[ h(x,y,\kappa) = y^2 + x^3 + \kappa x. \]
Then, there is a linear change of coordinates $(\xi_u,\xi_v) = (a_{11} y_1 + a_{12} y_2, a_{21} y_1 + a_{22} y_2)$, $(\chi_u,\chi_v) = (b_{11} x_1 + b_{12} x_2, b_{21} x_1 + b_{22} x_2)$, and $k-k^* = 2^{-14/9} (4\cdot2^{2/3}-1)^{1/3} \kappa \approx 0.594984 \,\kappa$, such that in the new coordinates we have
\begin{align*}
\widetilde H_+ &= y_1^2 - y_2^2 + x_1^3 - 3 x_1 x_2^2 + \kappa x_1 = \Re h(x,y,\kappa), \\
\widetilde H_- &= 2 y_1 y_2 + 3 x_1^2 x_2 - x_2^3 +\kappa x_2 
= \Im h(x, y, \kappa),
\end{align*} 
that is,
\begin{equation*}
h = \widetilde H_+ + i \widetilde H_-.
\end{equation*}
The numbers $a_{ij}, b_{ij}$ in the linear change of coordinates are explicitly given as complicated radical expressions. We report here their approximate values:
\begin{gather*}
a_{11} = -a_{22} \approx 1.49505, \ a_{21} = -a_{12} \approx 0.592494, \\
b_{11} = -b_{21} \approx 1.1239, \ b_{12} = b_{22} \approx 0.485921.
\end{gather*}

Thus, the fibration induced by $\widetilde{\mathcal F} = (\widetilde H_+, \widetilde H_-, \widehat L)$ on $\R^5$ and the fibration induced by $f$ on $\C^2\times\R$ are isomorphic.
Specifically, the fiber 
\[ \widetilde{\mathcal F}^{-1}\bigl(\varepsilon_1, \, \varepsilon_2, \, \kappa \bigr) \]
is, up to a linear coordinate transformation, exactly the fiber $f^{-1}(\varepsilon, \kappa)$, with $\varepsilon = \varepsilon_1 + i \varepsilon_2$.
Both these fibrations near $0$ are locally isomorphic to the fibration induced by $\widehat{\mathcal F} = (\widehat H_+, \widehat H_-, \widehat L)$ on $\S^2_{*,u} \times \S^2_{*,v} \times [-2,\infty)$ in a neighborhood of $c_*$. 
\end{proof}

We obtain as a corollary the following proposition. We emphasize that $c_*$ is a degenerate singularity and thus is not included in the classification given by Eliasson~\cite{eliasson1984hamiltonian}.

\begin{proposition}\label{prop.Topology}\label{prop:a2-sphere}
The reduced fiber $\widehat F^{-1}(\mathfrak c_*) = \widehat{\mathcal F}^{-1}(0)$ is topologically a sphere $\S^2$ with one singular point $c_*$, which is an $A_2$ singularity.
\end{proposition}

\begin{proof}
By Theorem~\ref{Th:local}, we have that for a small ball $B'$ around the singularity $c_*$ and for $(h_1,h_2,k)$ sufficiently close to the critical value $\mathfrak c_*$, the topology of $\widehat F^{-1}(h_1,h_2,k) \cap B'$ coincides with the one of $f$.
Therefore, it suffices to analyze the topology of the fibers of $f$.

For $\kappa=0$, the singularity $x=y=0$ of $f$ is the standard $A_2$ singularity, and $f^{-1}(0)$ is a Riemann surface which is homeomorphic to $\C$. 
In particular, if we consider a small ball $B = \{|x|^2 + |y|^2 \le \rho\} \subseteq \C^2$, then $f^{-1}(0) \cap B$ is homeomorphic to a two-dimensional disk which is singular at the origin. 

Moreover, $\widehat F^{-1}(h_1,h_2,k)$ for $(h_1,h_2,k)$ sufficiently close to, and including, the central critical value $\mathfrak c_*$, is a fibration outside $B'$ and, therefore,  $\widehat F^{-1}(\mathfrak c_*) \setminus B'$ is also an open disk.
Therefore, the singular fiber $\widehat F^{-1}(\mathfrak c_*)$ is the union of two disks, the disk $\widehat F^{-1}(\mathfrak c_*) \cap B'$ containing the singular point $c_*$, and the disk $\widehat F^{-1}(\mathfrak c_*) \setminus B'$. The two disks are glued along a common boundary on $B'$ to form a topological $\S^2$ with one singular point $c_*$.
\end{proof}

\begin{remark}\label{rem:fibers}
We can further develop the analysis in the proof of Theorem~\ref{Th:local} to recover the topology of the fibers corresponding to values near $\mathfrak c_*$. 
For $\kappa=0$, $|\varepsilon|$ sufficiently small and $\varepsilon \ne 0$, $f^{-1}(\varepsilon) \cap B$ is a two-dimensional torus from which an open disk has been deleted; as well as $f^{-1}(0) \cap B$ for $\kappa\neq 0$.
For $\kappa\neq 0$ and $\varepsilon\neq 0$, two cases arise. The topology of $f^{-1}(\varepsilon)\cap B$ is a two-dimensional torus with an open disk removed when the polynomial $x^3+\kappa x-\varepsilon$ has three simple roots, and a two-dimensional pinched torus with an open disk removed when the polynomial has a double root.
The latter case occurs for $(\varepsilon,\kappa)$ on the discriminant locus $\Delta$ of the polynomial.
\end{remark}

\begin{remark}
We note here that our use of the term ``$A_2$ singularity'' follows the classification of singularities in, e.g., \cite{Arnoldbook} and differs from the terminology in \cite{Kudryavtseva2025} where ``$A_2$'' corresponds to a cuspidal singularity. In both cases, the $A_2$ singularity corresponds to the equation $x^2 + y^3 = 0$, but $x, y \in \C$ in this work, while $x, y \in \R$ in \cite{Kudryavtseva2025}.
\end{remark}

\subsection{Picard-Lefschetz monodromy}
\label{sec:pl-monodromy}

In the previous section, we showed that the local fibration around the central singularity induced by the reduced system is locally topologically equivalent to the local fibration around $0$ given by the function $f : \C^2 \times \R \to \C \times \R$ with
\begin{equation}
\label{eq:A2unfolding-f}
f(x,y,\kappa) = (y^2+x^3+\kappa x, \kappa),
\end{equation}
and with fibers $f^{-1}(\varepsilon, \kappa)$ given by 
\begin{equation}
\label{eq:A2unfolding}
y^2 + x^3 + \kappa x = \varepsilon.
\end{equation}
Equation~\eqref{eq:A2unfolding} defines a family of Riemann surfaces, $R_{\varepsilon,\kappa}$, with $\varepsilon\in\C$ and $\kappa\in\R$. 

As explained in the previous section, these Riemann surfaces are singular for the values of $(\varepsilon,\kappa)\in \C\times\R$ on the discriminant locus $\Delta$ of the polynomial
\begin{equation}
\label{eq:x}
x^3+\kappa x-\varepsilon,
\end{equation}
given by the equation
\begin{equation*}\label{eq:Delta}
\varepsilon^2=-\frac{4}{27}\kappa^3.
\end{equation*}

The discriminant locus $\Delta$ of the polynomial $x^3 + \kappa x - \varepsilon$ consists of four threads, $\lambda_j$, $j=1,\dots,4$, see Figure~\ref{fig:local-threads}.
The threads $\lambda_1$ and $\lambda_2$ form a cusp at $0$ and are respectively parameterized as $\varepsilon=-i(4\kappa^3/27)^{1/2}$ and $\varepsilon=i(4\kappa^3/27)^{1/2}$ with $\kappa > 0$.
Since $\varepsilon \in i\R$ on $\lambda_{1,2}$ we refer to them as imaginary threads.
The real threads $\lambda_3$ and $\lambda_4$ are respectively parameterized as $\varepsilon = (4|\kappa|^3/27)^{1/2}$ and $\varepsilon = -(4|\kappa|^3/27)^{1/2}$ with $\kappa < 0$, and similarly form a cusp at $0$.

\begin{figure}[ht!]
\centering
\hfil
\includegraphics[width=0.42\linewidth]{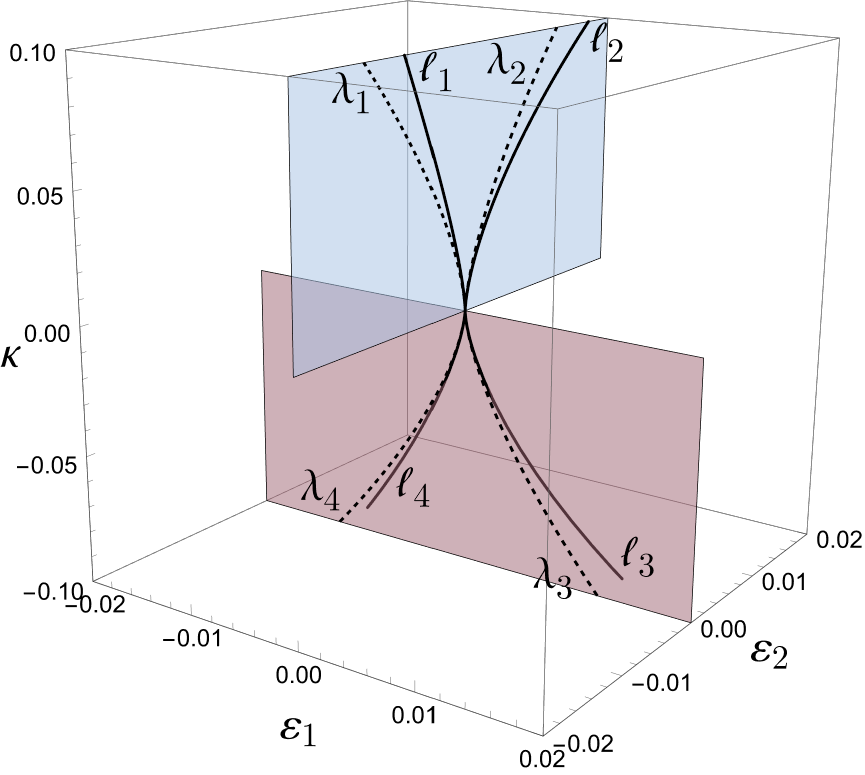}
\hfil
\includegraphics[width=0.42\linewidth]{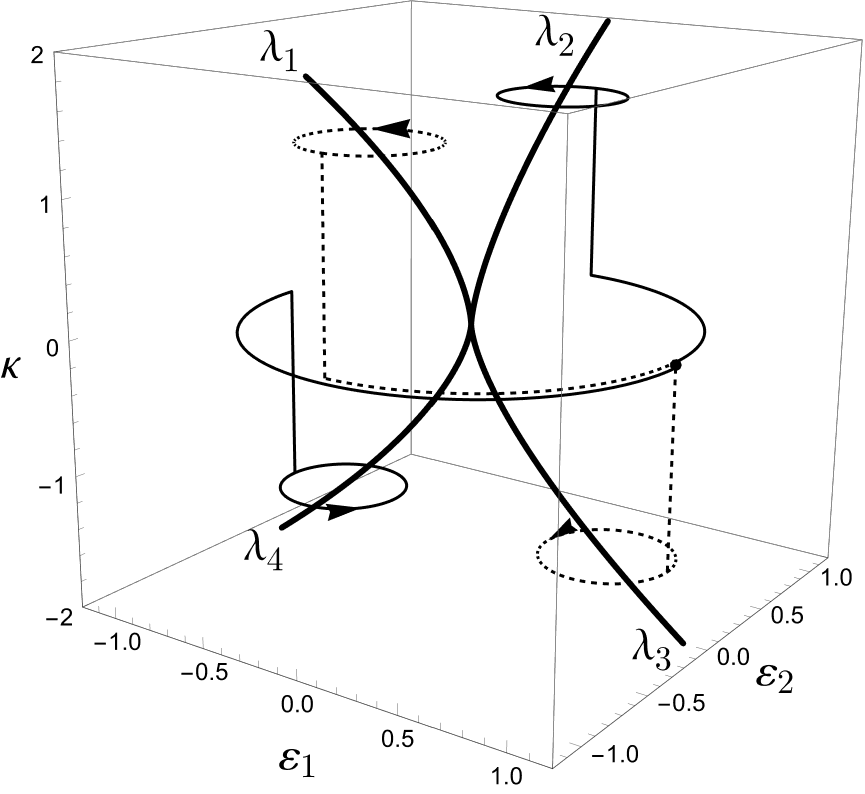}
\hfil
\caption{(a) Threads $\lambda_j$, $j=1,\dots,4$, comprising the discriminant locus $\Delta$ of $x^3+\kappa x-\varepsilon$, and image under the map $\psi$ of the threads $\ell_j$, $j=1,\dots,4$, comprising the set of critical values of $\mathcal F$. (b) Paths used for the computation of Picard-Lefschetz monodromy. The paths around the threads $\lambda_{1,2}$ (resp., $\lambda_{3,4}$) have been shifted slightly upward (resp., downward) for a clearer visualization.}
\label{fig:local-threads}
\label{fig:pl-paths}
\end{figure}

The map $\phi: \R^3 \to \C\times\R$ in Theorem~\ref{Th:local} sends the set of critical values of $\mathcal F$ in a neighborhood of $0$ to the discriminant locus $\Delta$ of the polynomial \eqref{eq:x} in a neighborhood of $0$.
In particular, it maps each thread $\ell_j$, $j=1,\dots,4$, to the thread $\lambda_j$ of $\Delta$. Figure~\ref{fig:local-threads}(a) shows the set of critical values near $0$ of the integral map $\mathcal F$, and the discriminant locus $\Delta$ in the space $(\varepsilon_1,\varepsilon_2,\kappa)$ near $0$.

We now proceed to compute the Picard-Lefschetz monodromy of the fibration given by the Riemann surfaces $R_{\varepsilon,\kappa}$ over the base space $\C \times \R$.
We fix a base point $(\varepsilon,\kappa)=(1,0)$ in the base space. Then the fundamental group $\pi_1(\Delta^c,(1,0))$ of the complement of $\Delta$ with base point $(1,0)$,  acts on the first homology group of the fibers $R_{\varepsilon,\kappa}$, $(\varepsilon,\kappa) \not\in \Delta$.
The action $N(a)$ of the fundamental group on a cycle $a$ is determined via the Picard-Lefschetz formula~\cite{zolkadek2006monodromy}
\begin{equation}\label{eq:PL}
N(a)=a-(a,\delta)\delta,
\end{equation}
where $\delta$ is a vanishing cycle.

The Riemann surfaces $R_{\varepsilon,\kappa}$, for $(\varepsilon,\kappa)\not\in\Delta$, are given by a two-sheet covering (two values of $y$), with three ramification points given by the roots of \eqref{eq:x}.
The first homology group of the fibers $R_{\varepsilon,\kappa}$ is two-dimensional and its bases and monodromy can be analyzed using the roots of \eqref{eq:x}.
For $\kappa=0$, $\varepsilon=1$, let $x_j=e^{2\pi i j /3}$, $j=0,1,2$, be the corresponding roots of \eqref{eq:x}. 
Let $\alpha_j$, $j=0,1,2$ be the cycles in $R_{0,1}$, obtained by lifting the segment from $x_{j-1}$ to $x_j$ in the \emph{upper sheet} and then backward from $x_j$ to $x_{j-1}$ in the \emph{lower sheet}.
Here and in the sequel, we use the convention that $x_{-1} = x_2$ and $\alpha_{-1} = \alpha_2$. 
The branching points of $R_{1,0}$ and the segments between them that give the cycles $\alpha_j$ are represented in Figure~\ref{fig:loops}.
Denoting by $(\alpha_i,\alpha_{i+1})$ the intersection number of the cycles $\alpha_i$, $\alpha_{i+1}$, we have
\begin{equation}
\label{eq:delta}
\alpha_0+\alpha_1+\alpha_2=0,\quad (\alpha_i,\alpha_{i+1})=1.
\end{equation}
To express monodromy as a matrix we work with the basis $\{\alpha_1,\alpha_0\}$. If the monodromy transformation along a path $\gamma$ is $N(\alpha_1) = n_{11} \alpha_1 + n_{12} \alpha_0$, $N(\alpha_0) = n_{21} \alpha_1 + n_{22} \alpha_0$, then we express the transformation as the matrix
\[ N(\gamma) = \begin{bmatrix} n_{11} & n_{12} \\ n_{21} & n_{22} \end{bmatrix}. \]

\begin{figure}[ht!]
\centering
\includegraphics[width=0.4\linewidth]{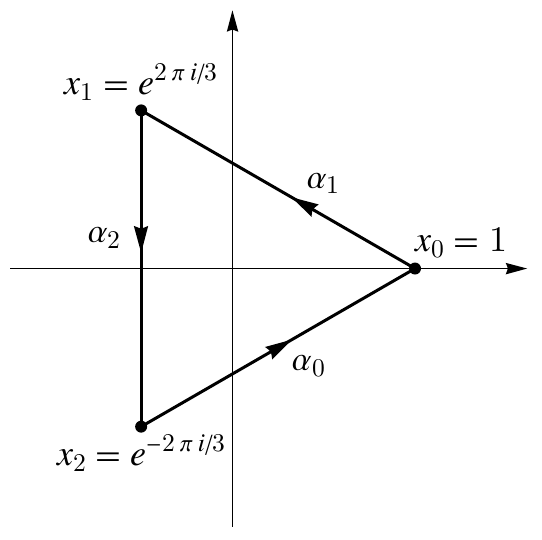}
\caption{Schematic representation of the upper sheet of the Riemann surface $R_{0,1}$. The black dots indicate the ramification points $x_k$. The segments that join the branching points and give the cycles $\alpha_k$ are represented by oriented straight lines.}
\label{fig:loops}
\end{figure}

Similarly to the discussion in Section~\ref{sec:monodromy-tavis-cummings}, Proposition~\ref{prop:fund-group-tc}, we have $\pi_1(\Delta^c,(1,0)) \simeq \Z * \Z * \Z$.
We compute Picard-Lefschetz monodromy along representatives $\hat\gamma_j$, $j=1,\dots,4$, of homotopy classes, such that each $\hat\gamma_j$ winds once around the thread $\lambda_j$ and is homotopically equivalent to the path $\psi \circ \gamma_j$, where $\gamma_j$ is the path used to compute Hamiltonian monodromy in Section~\ref{sec:monodromy-tavis-cummings}.

We start by analyzing the monodromy around the real threads $\lambda_3$, $\lambda_4$.
Let $\hat\gamma_3$ be the path starting at $(\varepsilon,\kappa)=(1,0)$, decreasing the value of $\kappa$ until it arrives close to $\lambda_3$, then positively winding in the complex $\varepsilon$-plane and back to the starting point, see Figure~\ref{fig:pl-paths}(b). 

For real parameters $\varepsilon$, $\kappa$, inside the cusp region there are three real roots of \eqref{eq:x} and outside only one.
This means that along this first path \eqref{eq:x} has one real and two complex conjugate roots.
The complex conjugate roots are obtained by transporting the roots $x_1$ and $x_2$, while the root $x_0$ is not affected.
The complex conjugate roots coalesce when touching the cusp curve implying that the cycle $\alpha_2=-\alpha_0-\alpha_1$ is the vanishing cycle along this path. 
Hence, using the Picard-Lefschetz formula, Eq.~\eqref{eq:PL}, and the intersection numbers, Eq.~\eqref{eq:delta}, we obtain $N_3(\alpha_0)=\alpha_0+\alpha_2=-\alpha_1$, and $N_3(\alpha_1)=\alpha_1-\alpha_2=\alpha_0+2\alpha_1$.

The general approach for dealing with the monodromy around the other threads is similar, but we must first transport the roots $x_j$ (and hence the cycles $\alpha_j$) around a convenient path, as shown in Figure~\ref{fig:pl-paths}(b). 

To calculate the monodromy transformation $N_4$ along a path $\hat\gamma_4$ around the real thread $\lambda_4$, we first consider a half-circle (in the negative direction) in the complex $\varepsilon$-plane, $\kappa=0$. It affects the roots $x_j$ by rotation by $-\pi/3$, that is, the roots are transformed to $x'_i = x_i e^{-\pi i/3}$ and the cycles $\alpha_i$ are transformed to $\alpha'_i$ from $x'_{i-1}$ to $x'_{i}$. 
We then decrease $\kappa$ until arriving close to the lower cusp thread in the real $(\varepsilon,\kappa)$-plane, then winding positively in the complex $\varepsilon$-plane and back to the starting point. 
The vanishing cycle is now $\alpha'_1$.
Then, the Picard-Lefschetz formula, Eq.~\eqref{eq:PL}, and the intersection numbers, Eq.~\eqref{eq:delta}, give that $N_4(\alpha_0)=\alpha_0-\alpha_1$ and $N_4(\alpha_2)=\alpha_2+\alpha_1$. Since $\alpha_1=-\alpha_0-\alpha_2$, we obtain $N_4(\alpha_1)=\alpha_1$. 

We then consider the monodromy along a path $\hat\gamma_2$ around the imaginary thread $\lambda_2$. To obtain the purely imaginary part of $\varepsilon$, we have to wind by $\pi/2$ in the complex $\varepsilon$-plane, for $\kappa=0$. 
We have
\begin{equation}\label{eq:x'}
{x^\prime}^3+\kappa x^\prime=i\varepsilon.
\end{equation}
Hence, the roots $x_j$ are modified to $x'_j = x_j e^{\pi i / 6}$. 
We want to obtain a purely real equation. Substituting $x'=-i\xi$, in \eqref{eq:x'} and dividing by $i$, we get $\xi^3-\kappa\xi=\varepsilon$.
The roots $\xi_j$,  obtained from $x_j$, $j=0,1,2$, are $\xi_j=x_{j+1}$ and the cycles $\alpha'_j$ as previously go from ${\xi}_{j-1}$ to ${\xi}_j$. 
We then increase $\kappa$ until arriving close to the upper cusp thread in the plane $i\R\times\R$, then winding positively in the complex $\varepsilon$-plane and back to the starting point. 
The vanishing cycle is $\alpha'_1$.
This gives 
$N_2(\alpha_1)=\alpha_1$ and $N_2(\alpha_0)=\alpha_0-\alpha_1$.

We repeat the same procedure for $\lambda_1$ considering a path $\hat\gamma_1$. 
We first wind by $-\pi/2$ in $\varepsilon$. 
Thus, the roots $x_j$ wind by $-\pi/6$, and we get the roots $x'_j = x_j e^{-\pi i/6}$. We make the same change of variable as in the previous case, and dividing by $i$, we get $\xi^3-\kappa\xi=-\varepsilon$, obtaining that $x_0$ is transformed to $\xi_0=e^{\pi i/3}$, $\xi_1=-1$, and $\xi_2=e^{-\pi i/3}$. 
We get that the cycle $\alpha'_0$ joining $\xi_2$ to $\xi_0$ is the vanishing cycle when we reach the imaginary thread $\lambda_1$.
This gives the monodromy $N_1(\alpha_0)=\alpha_0$, $N_1(\alpha_1)=\alpha_1+\alpha_0$.

We finally summarize the results in the following proposition using the conventions described above. 

\begin{proposition} The generators of the Picard-Lefschetz monodromy of the reduced Tavis-Cummings system in a neighborhood of the central $A_2$ singularity are given by:
\begin{gather*}
N(\hat\gamma_1) = \begin{bmatrix}
  1 & 1 \\
 0 & 1
\end{bmatrix}, \ 
N(\hat\gamma_2) = N(\hat\gamma_4) =
\begin{bmatrix}
  1 & 0 \\
 -1 & 1
\end{bmatrix}, \ 
N(\hat\gamma_3) =\begin{bmatrix}
 2 & 1 \\
 -1 & 0 
\end{bmatrix}
.
\end{gather*}
\end{proposition}

We stress that we obtain the lower right $2\times2$ block of the matrices obtained in Section~\ref{sec:monodromy-tavis-cummings}. Note that in the full space, there is a global $\S^1$ action. Hence, the additional cycle (given by the $\S^1$ action) in the full space is not affected by the monodromy. 

\section{Conclusion and perspectives}
\label{sec:conclusions}

In this work, we investigated the topological structure of the singular Lagrangian fibration of a special case of the two-spin Tavis–Cummings (TC) system. 
For the specific choice of parameters defining the special Tavis-Cummings system, we have identified what is, to the best of our knowledge, the first physical example of a completely integrable Hamiltonian system exhibiting an $A_2$ singularity. 
Moreover, we described the topology of all singular fibers in the interior of the image of the integral map and showed that the Lagrangian fibration by three-dimensional tori over the regular values of the integral map is not globally trivializable by explicitly computing the monodromy matrices of the system.
We further demonstrated the correspondence between the Hamiltonian monodromy of the system and the Picard-Lefschetz monodromy of the versal unfolding of the $A_2$ singularity.
We conclude this paper with some additional perspectives on the Lax pair formalism and the appearance of the $A_2$ singularity in the STC system.

A distinctive feature of the TC system is that, in any number of degrees of freedom, it admits a description in terms of a spectral Lax pair. As discussed in Appendix~\ref{sec:lax}, it is precisely through a Lax pair approach that we have identified the parameter values corresponding to the special TC system. In a forthcoming paper, we will study how to recover the bifurcation diagram and analyze Hamiltonian monodromy using the spectral Lax pair of the system, as this algebraic approach appears to provide a systematic method for addressing these questions~\cite{audin2002hamiltonian, Babelon2003, Gutierrez-Guillen2024}.

The fine-tuning of the parameters required for the appearance of the $A_2$ singularity raises the question of how the system changes under small parameter variations.
Recall that to define the STC system we chose specific values of the parameters $\delta_1=1/2$, $\delta_2=3/2$, $\omega=1$, satisfying the resonance condition $\delta_1 + \delta_2 = 2\omega$.
Fixing parameter values allowed us to draw the sets of critical values in Figure~\ref{fig:stcm-bd} and to numerically compute the Hamiltonian monodromy of the system.
However, the proof of Proposition~\ref{prop:rank-1-parameterization-stcm} and the discussion in Appendix~\ref{sec:lax}, show that the TC system has an $A_2$ singularity for a wide range of parameters $\delta_1, \delta_2 > 0$ satisfying the resonance condition $\delta_1 + \delta_2 = 2\omega$ with $\delta_1 \ne \delta_2$.
If the resonance condition is not satisfied, then the system does not have an $A_2$ singularity. We have checked that if we keep $\delta_1=1/2$, $\delta_2=3/2$ fixed, and detune $\omega$, then the four threads $\ell_1,\dots,\ell_4$ do not meet at a central critical value. For $\omega$ slightly larger than $1$, the threads $\ell_1$ and $\ell_4$ combine to form a smooth thread of rank-$1$ non-degenerate focus-focus-regular singularities, while $\ell_2$ similarly combines with $\ell_3$. The new combined threads do not intersect. For $\omega$ slightly smaller than $1$, $\ell_1$ combines with $\ell_3$, while $\ell_2$ combines with $\ell_4$. That is, a small detuning of the resonance condition replaces the four threads meeting at the $A_2$ singularity with two smooth non-intersecting threads. 

Recall that the Jaynes-Cummings (JC) system, that is, the TC system with $N=1$ spins, has an $A_1$ (focus-focus) singularity, which is associated with non-trivial Hamiltonian monodromy. 
It is then natural to ask if higher-order $A_N$ singularities with $N \ge 2$ can also appear in integrable Hamiltonian systems with compact singular Lagrangian fibrations, and what are their implications. 
In this paper, we answered these questions for the case $N=2$, as we showed that the $A_2$ singularity appears in the TC system with $N=2$ spins.

We note that there is no general method for constructing an integrable Hamiltonian system with an $A_2$ singularity having a compact singular Lagrangian fibration given by an explicit integral map. 
An attempt to construct such a system by applying Flaschka's method to the holomorphic function $f:\C^2 \to \C$~\cite{Flaschka1988} with $f(x,y) = y^2 + x^3$, corresponding to the $A_2$ singularity, produces a 2-DOF integrable Hamiltonian system with non-complete flows and non-compact fibers.
In particular, the regular fibers are two-dimensional tori with a disk removed---these are the Riemann surfaces described in Section~\ref{sec:pl-monodromy}.
Therefore, the STC system provides a compactification of this system in an algebraically closed form for the integral map, while at the same time being physically significant.

The question of the existence of integrable Hamiltonian systems with an $A_N$ singularity having a compact singular Lagrangian fibration given by an explicit integral map remains open for $N \ge 3$. 
We conjecture that the TC system with $N$ spins may provide such examples, based on the observation that the existence of the $A_1$ singularity in the JC system and of the $A_2$ singularity in the two-spin TC system is related to the spectral Lax pair formalism and the degeneracy of the roots of the spectral curve.
Recall that for the TC system with $N$ spins, the Lax pair approach \cite{Babelon2012, Gutierrez-Guillen2024}, associates to the TC system a complex polynomial $Q_{2N+2}(\lambda)$ of degree $2N+2$, with real coefficients. 
The \emph{spectral curve} is the Riemann surface defined by $\mu^2 = Q_{2N+2}(\lambda)$.
In the case $N=1$, the polynomial $Q_4(\lambda)$ has, for specific choices of parameters, two complex conjugate roots with multiplicity $2$ which correspond to the $A_1$ singularity of the JC system. 
Similarly, in the case $N=2$, the polynomial $Q_6(\lambda)$ has, for specific choices of parameters, two complex conjugate roots with multiplicity $3$ which, as we showed in this paper, correspond to an $A_2$ singularity.
In general, for $N$ spins we may expect that the polynomial $Q_{2N+2}(\lambda)$ will have for specific choices of parameters two complex conjugate roots $\lambda_*$, $\bar\lambda_*$ with multiplicity $N+1$.
That is, the corresponding Riemann surface $\mu^2 = Q_{2N+2}(\lambda)$ will have an $A_N$ singularity at $\lambda_*$, since the local normal form around this singularity is $\mu^2=\lambda^{N+1}$, and, for the same reason, an $A_N$ singularity at $\bar\lambda_*$.
This raises the possibility that the corresponding integral map also has a highly degenerate $A_N$ singularity.
We would call a TC system with $N$ spins having this property a \emph{special Tavis-Cummings system of order $N$} (STC$_N$ system). 
Based on these observations, we formulate the following conjecture, generalizing what is known for $N=1,2$:

\begin{conjecture*}
For each $N \ge 1$, there exists a special Tavis-Cummings system of order $N$, that is, there exist parameter configurations such that the polynomial $Q_{2N+2}(\lambda)$ for the TC system with $N$ spins has two complex conjugate roots with multiplicity $N+1$, and for these parameter configurations the integral map has a highly degenerate $A_N$ singularity. 
\end{conjecture*}

\section*{Acknowledgments}
The second, third, and fourth author have been supported by the EIPHI Graduate school (contract ``ANR-17-EURE-0002'') and by the Région Bourgogne Franche-Comté.
The third author also benefited from the support of the IRL 2001 Solomon Lefschetz CNRS-UNAM,  Croatian Science Foundation (HRZZ) grant IP-2022-10-9820, partial support by the Horizon grant 101183111-DSYREKI-HORIZON-MSCA-2023-SE-01.

\appendix

\section{Reduction of the \texorpdfstring{$\S^1$}{S1} action}
\label{sec:algebraic-reduction}

In this section we describe in detail the reduction of the $\S^1$ action $\varphi_K$ \eqref{eq:S1-action} in terms of algebraic invariants. 
The $\S^1$ action has the linear invariants $u_3$, $v_3$ and the quadratic invariants $|z|^2$, $|u|^2 = 1 - u_3^2$, $|v|^2 = 1 - v_3^2$, $u \bar{v}$, $u \bar{z}$, and $v \bar{z}$. 
All polynomial invariants of $\varphi_K$ factor through the given linear and quadratic invariants. 
This can be shown by considering the diagonal action of $\varphi_K$ on the space of monomials in variables $u_3, v_3, u, v, z$, given by
\[ u_3^{n_1} v_3^{n_2} u^{n_3} \bar{u}^{n_4} u^{n_5} \bar{u}^{n_6} z^{n_7} \bar{z}^{n_8}
\xrightarrow{\ \varphi_K^t\ } e^{i(n_3+n_5+n_7-n_4-n_6-n_8) t} u_3^{n_1} v_3^{n_2} u^{n_3} \bar{u}^{n_4} v^{n_5} \bar{v}^{n_6} z^{n_7} \bar{z}^{n_8}. \]
For a monomial to be invariant under the $\S^1$ action we must have $n_3+n_5+n_7=n_4+n_6+n_8$ and, therefore, each of the factors $u, v, z$ must be paired with exactly one of $\bar{u}, \bar{v}, \bar{z}$.
Then, a theorem of Schwarz ensures that smooth functions invariant under the $\S^1$ action also factor through the same invariants \cite{Schwarz1975}.

For the algebraic reduction we use here the following real invariants which are combinations of the linear and quadratic invariants given earlier:
\begin{equation}
\label{eq:def-algebraic-invariants}
u_3, \ v_3,\ K, \
X_1 + i Y_1 = v \bar z, \
X_2 + i Y_2 = z \bar u, \
X_3 + i Y_3 = u \bar v.
\end{equation}
Expressed in terms of real coordinates, we have
\begin{equation*}
\begin{aligned}
X_1 &= p v_1 + q v_2, & X_2 &= p u_1 + q u_2, & X_3 &= u_1 v_1 + u_2 v_2, \\
Y_1 &= -q v_1 + p v_2, & Y_2 &= q u_1 - p u_2, & Y_3 &= u_2 v_1 - u_1 v_2. \\
\end{aligned}
\end{equation*}
The dynamics in the reduced space $K^{-1}(k) / \S^1$ is then determined by the Poisson structure given in Table~\ref{tbl:poisson-algebraic-invariants}.

The algebraic invariants in Eq.~\eqref{eq:def-algebraic-invariants} are not independent. They satisfy the syzygies $\sigma_k = 0$, $k=1,\dots,9$, where
\begin{equation*}
\label{eq:syzygies-algebraic-invariants}
\begin{aligned}
\sigma_1 &= X_1^2 + Y_1^2 - 2 (1 - v_3^2) (K - u_3 - v_3), \\
\sigma_2 &= X_2^2 + Y_2^2 - 2 (1 - u_3^2) (K - u_3 - v_3), \\
\sigma_3 &= X_3^2 + Y_3^2 - (1 - v_3^2) (1 - u_3^2), \\
\sigma_4 &= X_1 X_3 - Y_1 Y_3 - (1 - v_3^2) X_2, \\
\sigma_5 &= X_1 Y_3 + X_3 Y_1 + (1 - v_3^2) Y_2, \\
\sigma_6 &= X_2 X_3 - Y_2 Y_3 - (1 - u_3^2) X_1, \\
\sigma_7 &= X_2 Y_3 + X_3 Y_2 + (1 - u_3^2) Y_1, \\
\sigma_8 &= X_1 X_2 - Y_1 Y_2 - 2 (K - u_3 - v_3) X_3, \\
\sigma_9 &= X_1 Y_2 + X_2 Y_1 + 2 (K - u_3 - v_3) Y_3.
\end{aligned}
\end{equation*}
Notice that the given syzygies are themselves not independent.

\begin{table}
\centering
\begin{tabular}{>{$}c<{$}|*{9}{>{$}c<{$}}}
 & K & u_3 & v_3 & X_1 & Y_1 & X_2 & Y_2 & X_3 & Y_3 \\
\midrule
\{K,\cdot\} & 0 & 0 & 0 & 0 & 0 & 0 & 0 & 0 & 0 \\
\{u_3,\cdot\} & 0 & 0 & 0 & 0 & 0 & -Y_2 & X_2 & Y_3 & -X_3 \\
\{v_3,\cdot\} & 0 & 0 & 0 & Y_1 & -X_1 & 0 & 0 & -Y_3 & X_3 \\
\{X_1,\cdot\} & 0 & 0 & -Y_1 & 0 & A & -Y_3 & -X_3 & -v_3 Y_2 & -v_3 X_2 \\
\{Y_1,\cdot\} & 0 & 0 & X_1 & -A & 0 & -X_3 & Y_3 & -v_3 X_2 & v_3 Y_2 \\
\{X_2,\cdot\} & 0 & Y_2 & 0 & Y_3 & X_3 & 0 & -B & u_3 Y_1 & u_3 X_1 \\
\{Y_2,\cdot\} & 0 & -X_2 & 0 & X_3 & -Y_3 & B & 0 & u_3 X_1 & -u_3 Y_1 \\
\{X_3,\cdot\} & 0 & -Y_3 & Y_3 & v_3 Y_2 & v_3 X_2 & -u_3 Y_1 & -u_3 X_1 & 0 & C \\
\{Y_3,\cdot\} & 0 & X_3 & -X_3 & v_3 X_2 & -v_3 Y_2 & -u_3 X_1 & u_3 Y_1 & -C & 0 \\
\end{tabular}
\caption{Poisson structure for the algebraic invariants in \eqref{eq:def-algebraic-invariants}. In the table, $A = 2 v_3 (K-u_3-v_3) + 1-v_3^2$, $B = 2 u_3 (K-u_3-v_3)+1-u_3^2$, and $C = u_3 (1-v_3^2)-(1-u_3^2) v_3$.}
\label{tbl:poisson-algebraic-invariants}
\end{table}

Finally, we note that the functions $H_1$ and $H_2$ can be expressed in terms of the $\S^1$ algebraic invariants as
\begin{align*}
\widehat H_1 &= (\delta_1-\omega) u_3 + \sqrt{2} g Y_2 - \frac{2 g^2}{\delta_2 - \delta_1} (X_3 + u_3 v_3), \\
\widehat H_2 &= (\delta_2-\omega) v_3 - \sqrt{2} g Y_1 + \frac{2 g^2}{\delta_2 - \delta_1} (X_3 + u_3 v_3).
\end{align*}

Recall that the reduced spaces $K^{-1}(k) / \S^1$ are smooth manifolds for $k \in (-2,0) \cup (0,2) \cup (2,\infty)$. We have the following characterization.

\begin{proposition}\label{prop:reduced-spaces}
The reduced space $K^{-1}(k) / \S^1$ is diffeomorphic to $\CP^2$, for $k \in (-2,0)$. It is diffeomorphic to a symplectic blow-up of $\CP^2$ on two disjoint balls, or---equivalently---the symplectic blow-up of $\S^2 \times \S^2$ on a ball, for $k \in (0,2)$. Finally, it is diffeomorphic to $\S^2 \times \S^2$, for $k \in (2,\infty)$.
\end{proposition}

\begin{proof}
To prove this statement we exhibit the manifolds $K^{-1}(k) / \S^1$ as compact toric manifolds and we compute the corresponding Delzant polytopes \cite{Delzant1988}.
In particular, consider the Hamiltonian $\T^3$ action $\Psi$ on $M$ given by
\[ \Psi((e^{it}, e^{is_1}, e^{is_2}), (u,u_3;v,v_3;z))
= \varphi_K^t \circ \varphi_{J_1}^{s_1} \circ \varphi_{J_2}^{s_2} (u,u_3;v,v_3;z) 
= (e^{i(t+s_1)}\textcolor{red}{u}, u_3; e^{i(t+s_2)} v, v_3; e^{it} z), \]
where $J_1=u_3$, $J_2=v_3$.
The action $\Psi$ induces an effective $\T^2$ action $\psi$ on the reduced space $K^{-1}(k) / \S^1$. In terms of algebraic invariants $|z|^2, u_3, v_3, u \bar{v}, u \bar{z}, v \bar{z}$, the $\T^2$ action $\psi$ is expressed as
\begin{equation}
\label{eq:psi}    
\psi((e^{is_1}, e^{is_2}), (|z|^2, u_3, v_3, u \bar{v}, u \bar{z}, v \bar{z}))
= (|z|^2, u_3, v_3, e^{i(s_1-s_2)} u \bar{v}, e^{is_1} u \bar{z}, e^{is_2} v \bar{z}).
\end{equation}
The Delzant polytope $D_k$ for the action $\psi$ on $K^{-1}(k) / \S^1$ can be determined as the convex hull of its vertices; the latter correspond to fixed points of $\psi$. 
The expression in Eq.~\eqref{eq:psi} implies that the fixed points of $\psi$ satisfy the relations $u \bar v = 0$, $u \bar z = 0$, and $v \bar z = 0$. 
Therefore, we can distinguish the following cases.

\begin{enumerate}[label={(\roman*)}]
\item If $u = v = 0$ then we find $J_1 = u_3 = \pm 1$ and $J_2 = v_3 = \pm 1$. 
Since on $K^{-1}(k)$ we have $|z|^2 = 2(k - u_3 - v_3)$ we must have $k \ge u_3 + v_3$. 
Therefore, the vertex $(J_1,J_2)=(-1,-1)$ can be attained for all $k \ge -2$, while the vertices $(J_1,J_2)=(-1,1)$ and $(J_1,J_2)=(1,-1)$ can be attained, for all $k \ge 0$, and the vertex $(J_1,J_2)=(1,1)$ can be attained, for all $k \ge 2$.

\item If $u = z = 0$, then $J_1 = u_3 = \pm 1$ and $J_2 = k - J_1 = k \mp 1$ with $|J_2| \le 1$.
Therefore, we obtain the vertices $(J_1,J_2) = (1,k-1)$, for $0 \le k \le 2$ and $(J_1,J_2) = (-1,k+1)$ for $-2 \le k \le 0$.

\item If $v = z = 0$, then $J_2 = v_3 = \pm 1$ and $J_1 = k - J_2 = k \mp 1$ with $|J_1| \le 1$.
Therefore, we obtain the vertices $(J_1,J_2) = (k-1,1)$, for $0 \le k \le 2$ and $(J_1,J_2) = (k+1,-1)$, for $-2 \le k \le 0$.
\end{enumerate}

To conclude, for $-2 \le k \le 0$ the Delzant polygon $D_k$ has three vertices, $(k+1,-1)$, $(-1,k+1)$, and $(-1,-1)$ and therefore in this case $K^{-1}(k)/\S^1$ is diffeomorphic to $\CP^2$.
For $k > 2$ the Delzant polygon $D_k$ has four vertices, $(\pm 1, \pm 1)$, and therefore in this case $K^{-1}(k)/\S^1$ is diffeomorphic to $\S^2 \times \S^2$.
Finally, for the case $0 \le k \le 2$ of intermediate values the Delzant polygon $D_k$ has five vertices, $(k-1,1)$, $(1,k-1)$, $(1,-1)$, $(-1,1)$, and $(-1,-1)$.
This polygon can be obtained from the square $[-1,1]^2$ by cutting off a corner.
This corresponds to a symplectic blowup of the space $\S^2 \times \S^2$ along a ball.
Alternatively, the same polygon can be obtained from the triangle with vertices $(k+1,-1)$, $(-1,k+1)$, and $(-1,-1)$, by cutting off two corners.
This corresponds to two symplectic blowups of the space $\CP^2$ along two disjoint balls.
\end{proof}

Notice that, the fact that $K^{-1}(k) / \S^1$ is diffeomorphic to $\CP^2$ when $-2 < k < 0$ also follows directly from the fact that $K$ is a definite positive function at $(u,u_3;v,v_3;z)=(0,-1;0,-1;0)$.
Therefore, each level set $K^{-1}(k)$ for $k$ greater than, but close to, $-2$ is diffeomorphic to $\S^5$.
It can then be checked that reducing the given $\S^1$ action on $\S^5$ gives $\CP^2$ as the reduced space.

Moreover, for $k > 2$ the map
\[ \S^2_u \times \S^2_v \to K^{-1}(k) \subseteq \S^2_u \times \S^2_v \times \C_* \colon 
(u, u_3; v, v_3) 
\mapsto \bigl(u, u_3; v, v_3; [2(k-u_3-v_3)]^{1/2}\bigr), 
\]
is a smooth global section for the $\S^1$ bundle induced by the $\S^1$ action on $K^{-1}(k)$.
This follows from the fact that for points $(u,u_3;v,v_3;z) \in K^{-1}(k)$, $k > 2$, we have $|z|^2 = 2(k-u_3-v_3) > 0$.
This also shows that $K^{-1}(k) \simeq \S^2_u \times \S^2_v \times \S^1$ and thus it gives an explicit diffeomorphism $K^{-1}(k) / \S^1 \simeq \S^2_u \times \S^2_v$.

\section{Rank-2 singularities}
\label{App.Rank2}

In this section we give the proof of Proposition~\ref{prop:rank-2-singularities} that describes the rank-$2$ singularities of the general TC system.
We find rank-$2$ singular points by considering the relation $X_K = x X_{H_1} + y X_{H_2}$ with $x, y \in \R$.

First, the relation $X_K(z) = x X_{H_1}(z) + y X_{H_2}(z)$ directly gives
\begin{align}
\label{eq:proof-rank-2:z}
z = \sqrt2 i g (x u + y v). 
\end{align}
Both relations $0 = X_K(u_3) = x X_{H_1}(u_3) + y X_{H_2}(u_3)$ and $0 = X_K(v_3) = x X_{H_1}(v_3) + y X_{H_2}(v_3)$ give after using \eqref{eq:proof-rank-2:z} the relation
\begin{align*}
i g^2 \Bigl[ \frac{1}{\delta_2 - \delta_1} (x-y) - x y \Bigr] (u \bar{v} - \bar{u} v) = 0.
\end{align*}
Assuming $u \bar{v} - \bar{u} v \ne 0$ we find
\begin{align}
\label{eq:proof-rank-2:relation-x-y}
x - y = (\delta_2 - \delta_1) x y.
\end{align}

Then, the relation $X_K(u) = x X_{H_1}(u) + y X_{H_2}(u)$, after using \eqref{eq:proof-rank-2:z}, gives
\begin{align*}
\Bigl[x (\delta_1-\omega) - 2 g^2 x^2 u_3 
- (x - y) \frac{2g^2}{\delta_2 - \delta_1} v_3 - 1 \Bigr] i u 
+ 2 i g^2 \Bigl[ \frac{1}{\delta_2 - \delta_1} (x-y) - x y\Bigr] u_3 v 
= 0,
\end{align*}
which, taking into account \eqref{eq:proof-rank-2:relation-x-y}, simplifies to
\begin{align*}
2 g^2 (x u_3 + y v_3) = (\delta_1-\omega) - \frac{1}{x}.
\end{align*}
Similarly, the relation $X_K(v) = x X_{H_1}(v) + y X_{H_2}(v)$ gives
\begin{align*}
2 g^2 (x u_3 + y v_3) = (\delta_2-\omega) - \frac{1}{y}.
\end{align*}
Let $C = x u_3 + y v_3$. Then, the previous two relations imply
\begin{align}
\label{eq:proof-rank-2:C}
C = x u_3 + y v_3 = \frac{1}{2g^2} (\delta_1-\omega) - \frac{1}{2g^2x}
= \frac{1}{2g^2} (\delta_2-\omega) - \frac{1}{2g^2y}.
\end{align}

Let $U = u \bar{v} + \bar{u} v$. Then, using \eqref{eq:proof-rank-2:z}, we obtain
\begin{gather*}
\label{eq:proof-rank-2:uz}
\frac{i g}{\sqrt2} (u \bar{z} - \bar{u} z)
= g^2 (2 x u \bar{u} + y (u \bar{v} + \bar{u} v) )
= 2 x g^2 (1-u_3^2) + y g^2 U, \\
\label{eq:proof-rank-2:vz}
\frac{i g}{\sqrt2} (v \bar{z} - \bar{v} z)
= g^2 (2 y v \bar{v} + x (u \bar{v} + \bar{u} v) )
= 2 y g^2 (1-v_3^2) + x g^2 U.
\end{gather*}
Substituting the previous two expressions into $H_1$ in \eqref{eq:H1-complex} and $H_2$ in \eqref{eq:H2-complex} we obtain.
\begin{align}
\label{eq:proof-rank-2:H1u}
H_1 &= (\delta_1-\omega) u_3 + 2 x g^2 (1-u_3^2) + y g^2 U
- \frac{g^2}{\delta_2 - \delta_1} (U + 2 u_3 v_3), \\
\label{eq:proof-rank-2:H2u}
H_2 &= (\delta_2-\omega) v_3 + 2 y g^2 (1-v_3^2) + x g^2 U
+ \frac{g^2}{\delta_2 - \delta_1} (U + 2 u_3 v_3).
\end{align}

Moreover, using again \eqref{eq:proof-rank-2:z}, the momentum $K$ becomes
\begin{align}
K 
& = u_3 + v_3 + \frac12 z \bar{z} \notag \\
& = u_3 + v_3 +  g^2 (x^2 u \bar{u} + y^2 v \bar{v} + x y (u \bar{v} + \bar{u} v)) \notag \\
\label{eq:proof-rank-2:K}
& = u_3 + v_3 + g^2 x^2 (1-u_3^2) + g^2 y^2 (1-v_3^2) + g^2 x y U.
\end{align}

Fixing the value of the momentum $K = k$, we can solve \eqref{eq:proof-rank-2:K} for $U$, and substitute into \eqref{eq:proof-rank-2:H1u} and \eqref{eq:proof-rank-2:H2u}.
Moreover, using \eqref{eq:proof-rank-2:C} to replace $\delta_1-\omega$ and $\delta_2-\omega$, and taking into account \eqref{eq:proof-rank-2:relation-x-y}, we obtain \eqref{eq:rank-2:H1} and \eqref{eq:rank-2:H2}, that is, 
\begin{align*}
(\delta_2 - \delta_1) H_1 &= -\frac{k}{x^2} - \frac{C^2 g^2}{x^2} + \frac{C}{x^2 y} + g^2 \Bigl( \frac{y^2}{x^2} + 2 \frac{x}{y} - 1 \Bigr),
\\
-(\delta_2 - \delta_1) H_2 &= -\frac{k}{y^2} - \frac{C^2 g^2}{y^2} + \frac{C}{x y^2} + g^2 \Bigl( \frac{x^2}{y^2} + 2 \frac{y}{x} - 1 \Bigr).
\end{align*}

We now want to determine bounds for the parameters $x$, $y$.
Define
\begin{align*}
a = a_1 + i a_2 = \frac{u}{z}, \ b = b_1 + i b_2 = \frac{v}{z}.
\end{align*}
Then 
\begin{align*}
|a|^2 = a_1^2 + a_2^2 = \frac{1-u_3^2}{2(K-u_3-v_3)}, \ 
|b|^2 = b_1^2 + b_2^2 = \frac{1-v_3^2}{2(K-u_3-v_3)}.
\end{align*}
Equation \eqref{eq:proof-rank-2:z} gives $\sqrt2 i g (x a + y b) = 1$ from which we obtain
\begin{align*}
y^2 |b|^2 = x^2 |a|^2 - \frac{\sqrt2 x}{g} a_2 + \frac{1}{2 g^2},
\end{align*}
and, by substituting $|a|^2$, $|b|^2$ and using $x u_3 + y v_3 = C$, we arrive at
\begin{align*}
a_2 = \frac{g}{\sqrt2 x} \Bigl( \frac{1}{2g^2} 
+ \frac{y}{2} \frac{C^2 - 2 C x u_3 + x^2 - y^2}{yK-C+(x-y)u_3} \Bigr).
\end{align*}
Then we have
\begin{align*}
a_1^2 = \frac{1-u_3^2}{2(K-u_3-v_3)} - a_2^2
= \frac{y}{2} \frac{1-u_3^2}{yK-C+(x-y)u_3} - \frac{g^2}{2 x^2} \Bigl( \frac{1}{2g^2} + \frac{y}{2} \frac{C^2 - 2 C x u_3 + x^2 - y^2}{yK-C+(x-y)u_3} \Bigr)^2.
\end{align*}
We finally obtain
\begin{align*}
a_1^2 = \frac{1}{8 g^2 x^2 (y K -C + (x-y) u_3)^2} P(x,y;u_3),
\end{align*}
where $P(x,y;u_3)$ is the following cubic polynomial in $u_3$:
\begin{equation}
\label{eq:proof-rank-2:def-P}
\begin{aligned}
P(x,y;u_3) & = 
-4 g^2 x^2 y (x-y) u_3^3 \\
& + (-4 C^2 g^4 x^2 y^2+8 C g^2 x^2 y-4 C g^2 x y^2-4 g^2 K x^2 y^2
-(x-y)^2) u_3^2 \\
& +2 (2 C^3 g^4 x y^2-3 C^2 g^2 x y+C^2 g^2 y^2+2 C g^4 x^3
   y^2-2 C g^4 x y^4+2 C g^2 K x y^2 \\
& \qquad +C x-C y+g^2 x^3 y-g^2 x^2 y^2+g^2 x y^3-g^2 y^4-K x y+K y^2 ) u_3 \\
& -C^4 g^4 y^2+2 C^3 g^2 y-2 C^2 g^4 x^2 y^2+2 C^2 g^4 y^4-2 C^2 g^2 K y^2 - C^2 -2 C g^2 x^2 y \\
& \qquad -2 C g^2 y^3+2 C K y-g^4 x^4 y^2+2 g^4 x^2 y^4-g^4 y^6+2 g^2 K x^2 y^2+2 g^2 K y^4 -K^2 y^2.
\end{aligned}
\end{equation}
We compute
\begin{align*}
P(x,y;-1) &= -(C + x - y - (K + g^2 (C + x)^2) y + g^2 y^3)^2 \le 0, \\
P(x,y;\phantom{-}1) &= -(C - x + y - (K + g^2 (C - x)^2) y + g^2 y^3)^2 \le 0.
\end{align*}
Therefore, $P(x,y;u_3)$ either has $0$ or $2$ roots in $[-1,1]$. Since $a_1^2 \ge 0$ and $P(x,y;\pm1) \le 0$, the allowed values of $x, y$ are those for which $P(x,y;u_3)$ has $2$ roots in $[-1,1]$. Since $y$ can be expressed in terms of $x$ as $y = x/(1+x(\delta_2 - \delta_1))$, we can determine those values of $x$ for which $Q(x;u_3) = P(x,x/(1+x(\delta_2 - \delta_1)); u_3)$ has $2$ roots in $[-1,1]$. These values form intervals and the boundaries of these intervals are determined by the conditions $Q(x;u_3) = d Q(x;u_3) / dx = 0$ with $u_3 \in [-1,1]$.

\section{Numerical computation of period lattices and Hamiltonian monodromy}
\label{sec:numerical-computation}

In this section we provide details of the numerical computation of the period lattices that were used to determine Hamiltonian monodromy of the special Tavis-Cummings system in Section~\ref{sec:monodromy-tavis-cummings}.

Recall that a vector $\vv{T} = (T_1, T_2, T_3) \in \R^3$ is a period vector on a regular fiber $F^{-1}(r)$ if 
\[ \varphi_{H_1}^{T_1} \circ \varphi_{H_2}^{T_2} \circ \varphi_{K}^{T_3} |_{F^{-1}(r)} = \mathrm{id}_{F^{-1}(r)}. \]
Since $X_K$ generates an $\S^1$ action with period $2 \pi$, the constant vector $\vv{T}_1 = (0,0,2\pi)$ is a period vector on every regular value $r$ of $F$. 
To compute a basis of the period lattice on $F^{-1}(r)$ we need to compute two more period vectors $\vv{T}_2(r)$ and $\vv{T}_3(r)$.

To compute period vectors we work as follows.
Given a regular value $r = (h_1,h_2,k)$ of $F$ we first compute a point $P \in F^{-1}(r)$ by numerically solving the equations $H_1(P) = h_1$, $H_2(P) = h_2$, $K(P) = k$. 
Then, we consider the function $m : \R^3 \to F^{-1}(r)$ given by 
\[ 
m(\vv{T}) = \varphi_{H_1}^{T_1} \circ \varphi_{H_2}^{T_2} \circ \varphi_{K}^{T_3} (P),
\]
with fixed $P \in F^{-1}(r)$ and we use the Newton-Raphson method to solve the three equations
\[ 
\theta_u(m(\vv{T})) - \theta_u(P) = 0, \ 
\theta_v(m(\vv{T})) - \theta_v(P) = 0, \ 
\phi(m(\vv{T})) - \phi(P) = 0, \ 
\]
given an initial guess for the period vector $\vv{T}$. 
Here, $(\phi, J)$ are symplectic polar coordinates on $\R^2_*$ defined by $p + i q = \sqrt{2J} \exp(i \phi)$.

First, we consider the regular value $r_0 = (2, 1, 1.8)$ of $F = (H_1, H_2, K)$. 
Then we compute that two period vectors which together with $\vv T_1 = (0, 0, 2\pi)$ form a basis of the period lattice at $r_0$ are given by 
\begin{equation}
\begin{aligned}
\vv{T}_2 & \approx (1.83862, 2.07173, -1.44104), \ 
\vv{T}_3 & \approx (-2.02757, -0.785264, 1.15808).
\end{aligned}
\end{equation}

At $k=1.8$ there are two focus-focus-regular points, which we denote by 
\[
\ff_1 \approx (-0.578466, 0.578466, 1.8), \  
\ff_2 \approx (-1.743013, 1.743013, 1.8), 
\]
while at $k=0.3$ there are two focus-focus-regular points, which we denote by
\[ 
\ff_3 \approx (1.793686, -1.164044, 0.3), \ 
\ff_4 \approx (1.164044, -1.793686, 0.3). 
\]
Then for each $j=1,2,3,4$ we consider a closed path $\gamma_j$ in the set $\mathcal R$ of regular values of $F$ which is based at $r_0$ and represents a homotopy class $[\gamma_j]$ in $\pi_1(\mathcal R, r_0)$. 
Specifically, the path $\gamma_j$ starts at $r_0$, moves along a straight line to the regular value $r_j = \ff_j + (0.5, 0, 0)$, circles the focus-focus point $\ff_j$ along the circle $C_j$ parameterized by $\ff_j + 0.5(\cos(2\pi s), \sin(2\pi s),0)$ with $s \in [0,1]$, and then moves from $r_j$ back to $r_0$ along a straight line.

For each closed path $\gamma_j$, $j=1,2,3,4$, we continue the period vectors $\vv{T}_2$, $\vv{T}_3$ along the path, and when we arrive back at $r_0$ we obtain the period vectors $\vv{T}^{(j)}_2$, $\vv{T}^{(j)}_3$.
We do not numerically continue $\vv{T}_1$ since it is constant.
The Hamiltonian monodromy matrix along the path $\gamma_j$ is given by
\[ M(\gamma_j) = 
\mathcal{B}(\vv{T}_1,\vv{T}_2^{(j)},\vv{T}_3^{(j)}) \,
\mathcal{B}(\vv{T}_1,\vv{T}_2,\vv{T}_3)^{-1}, \]
where $\mathcal{B}(\vv T_1, \vv T_2, \vv T_3)$ is the $3 \times 3$ matrix with rows given by the (row) vectors $\vv T_1$, $\vv T_2$, and $\vv T  _3$.

The numerical computations are approximate and they produce matrices which are close to matrices in $\SL(3,\Z)$.
In each case, we consider the nearest $\SL(3,\Z)$ matrix as the accurate expression for $M(\gamma_j)$.

Applying this procedure for the path $\gamma_1$ gives
\[ 
\vv T_2^{(1)} \approx (1.83862, 2.07173, -1.44104), \ 
\vv T_3^{(1)} \approx (-3.86619, -2.85699, 2.59913).
\]
The corresponding approximate monodromy matrix is computed as $\widetilde M(\gamma_1) = M(\gamma_1) + O(10^{-7})$,
with exact monodromy matrix 
\[ 
M(\gamma_1) 
= \begin{bmatrix}
 1 & 0 & 0 \\
 0 & 1 & 0 \\
 0 & -1 & 1 \\
\end{bmatrix}.
\]

\begin{remark}
We report the result of the numerical computation of the monodromy matrix in the form $M(\gamma) + O(\epsilon)$, where $M(\gamma) \in \SL(3,\Z)$ and $\epsilon \ll 1$. This may obscure that the main problem in the computation is not the small numerical error in the approximation of the period vectors but that during continuation the Newton-Raphson method might converge to a different period vector. If this happens, the produced approximate monodromy matrix will still have the form $M(\gamma) + O(\epsilon)$, where $M(\gamma) \in \SL(3,\Z)$ and $\epsilon \ll 1$, but $M(\gamma)$ will be incorrect. We control for this error by checking during continuation the size of the change of the period vector.
\end{remark}

Similarly, for the path $\gamma_2$ we compute
\[ 
\vv T_2^{(2)} \approx (1.64967, 3.35819, -8.00719), \ 
\vv T_3^{(2)} \approx (-1.83862, -2.07173,  7.72423).
\]
The corresponding approximate monodromy matrix is computed as $\widetilde M(\gamma_2) = M(\gamma_2) + O(10^{-6})$,
with exact monodromy matrix 
\[ 
M(\gamma_2) 
= \begin{bmatrix}
 1 & 0 & 0 \\
 -1 & 2 & 1 \\
 1 & -1 & 0 \\
\end{bmatrix}.
\]

Then, for the path $\gamma_3$ we compute
\[ 
\vv T_2^{(3)} \approx (-0.188951,  1.28646, -6.56615), \ 
\vv T_3^{(3)} \approx (-2.02757, -0.785264, 1.15808).
\]
The corresponding approximate monodromy matrix is computed as $\widetilde M(\gamma_3) = M(\gamma_3) + O(10^{-6})$,
with exact monodromy matrix 
\[ 
M(\gamma_3) 
= \begin{bmatrix}
 1 & 0 & 0 \\
 -1 & 1 & 1 \\
 0 & 0 & 1 \\
\end{bmatrix}.
\]

Finally, for the path $\gamma_4$ we compute
\[ 
\vv T_2^{(4)} \approx (1.64967, 3.35819, -8.00719), \ 
\vv T_3^{(4)} \approx (-1.83862, -2.07173, 7.72423).
\]
The corresponding approximate monodromy matrix is computed as $\widetilde M(\gamma_4) = M(\gamma_4) + O(10^{-6})$,
with exact monodromy matrix 
\[ 
M(\gamma_4) 
= \begin{bmatrix}
 1 & 0 & 0 \\
 -1 & 2 & 1 \\
 1 & -1 & 0 \\
\end{bmatrix}.
\]

The monodromy matrices reported in Section~\ref{sec:monodromy-tavis-cummings} are the same as the matrices reported here up to a change of basis. 
Specifically, consider the change of basis given by the matrix 
\[ A = \begin{bmatrix}
1 & 0 & 0 \\
1 & -1 & -1 \\
0 & 1 & 0    
\end{bmatrix}. \]
Then, redefining $M(\gamma_j)$ to be $A M(\gamma_j) A^{-1}$ we compute the reported
\begin{gather*}
M(\gamma_1) = \begin{bmatrix}
1 & 0 & 0 \\
0 & 1 & 1 \\
0 & 0 & 1 
\end{bmatrix}, \ 
M(\gamma_2) = M(\gamma_4) =
\begin{bmatrix}
 1 & 0 & 0 \\
 0 & 1 & 0 \\
 0 & -1 & 1 
\end{bmatrix}, \ 
M(\gamma_3) =
\begin{bmatrix}
 1 & 0 & 0 \\
 0 & 2 & 1 \\
 0 & -1 & 0
\end{bmatrix}.
\end{gather*}

\section{Spectral Lax pair}
\label{sec:lax}

In this section, we briefly describe the spectral Lax pair formalism associated with the two-spin TC, see \cite{Babelon2003, Babelon2012}.
We introduce the Lax matrices $\mathcal{L}$ and $\mathcal{M}$ depending on the spectral parameter $\lambda \in \C$ given by
\[
\mathcal{L}(\lambda) 
= \frac{1}{g^2}\Bigl(\lambda-\frac{\delta_1}{2}\Bigr)
\Bigl(\lambda-\frac{\delta_2}{2}\Bigr)
\bigl((2\lambda-\omega)\sigma_z + i\sqrt{2}g(\bar{z}\sigma_+-z\sigma_-)\bigr)
+\Bigl(\lambda-\frac{\delta_2}{2}\Bigr)\vv{u}\cdot\vv{\sigma}
+\Bigl(\lambda-\frac{\delta_1}{2}\Bigr)\vv{v}\cdot\vv{\sigma},
\]
and
\[
\mathcal{M}(\lambda)=-i\lambda\sigma_z+\frac{g}{\sqrt{2}}(\bar{z}\sigma_+-z\sigma_-),
\]
where $\vv{u}=(u_1,u_2,u_3)$, $\vv{v}=(v_1,v_2,v_3)$ and $\vv{\sigma}$ is the vector of Pauli matrices $(\sigma_x,\sigma_y,\sigma_z)$ with
\[
\sigma_x=\begin{bmatrix}
0 & 1 \\
1 & 0
\end{bmatrix},~\sigma_y=\begin{bmatrix}
0 & -i \\
i & 0
\end{bmatrix},~\sigma_z=\begin{bmatrix} 1 & 0 \\
0 & -1\end{bmatrix},~\sigma_+=\frac{1}{2}(\sigma_x+i\sigma_y),~\sigma_-=\frac{1}{2}(\sigma_x-i\sigma_y). 
\]

It can be shown that the matrices $\mathcal{L}$ and $\mathcal{M}$ satisfy the Lax equation $\dot{\mathcal{L}}(\lambda)=[\mathcal{M}(\lambda),\mathcal{L}(\lambda)]$ for any value of the complex parameter $\lambda$. Moreover, this equation is equivalent to the system's dynamical equations.

The characteristic polynomial of $\mathcal{L}(\lambda)$ does not depend on time and defines the \emph{spectral curve} given by the equation $\det (\mathcal{L}(\lambda)-\mu I)=0$.
For the two-spin TC, we obtain
\[
\mu^2=Q_6(\lambda)
\]
where $Q_6$ is a polynomial of degree $6$ in $\lambda$ that can be expressed as
\begin{align*}
Q_6(\lambda) &= \Bigl[\frac{(2\lambda-\omega)^2}{g^4}+\frac{4K}{g^2}\Bigr] 
\Bigl(\lambda-\frac{\delta_1}{2}\Bigr)^2
\Bigl(\lambda-\frac{\delta_2}{2}\Bigr)^2
+ \frac{2}{g^2}\Bigl(\lambda-\frac{\delta_1}{2}\Bigr)\Bigl(\lambda-\frac{\delta_2}{2}\Bigr)
\Bigl[H_1 \Bigl(\lambda-\frac{\delta_2}{2}\Bigr)+H_2\Bigl(\lambda-\frac{\delta_1}{2}\Bigr)\Bigr] \\
& \qquad + \Bigl(\lambda-\frac{\delta_1}{2}\Bigr)^2 
+ \Bigl(\lambda-\frac{\delta_2}{2}\Bigr)^2.
\end{align*}
Note that all the coefficients of $Q_6(\lambda)$ are constants of motion. 
 
The central critical value of the STC system occurs when the polynomial $Q_6$ has two complex conjugate triple roots. 
In this case, the polynomial $Q_6$ can be expressed as 
\[
Q_6(\lambda) = \frac{4}{g^4}(\lambda^2+a_1\lambda+a_0)^3,
\]
where $a_0$ and $a_1$ are two real parameters such that $a_1^2-4a_0 < 0$. Identifying the coefficients of the two expressions of $Q_6$ leads to six equations, from which the resonance condition $2\omega=\delta_1+\delta_2$, with $\delta_1 \ne \delta_2$, is obtained. 

In the STC system, we consider the parameters $\delta_1=\frac{1}{2}$, $\delta_2=\frac{3}{2}$, $\omega=1$ and $g=1$. We deduce that $a_1=-1$, $a_0=(4\cdot2^{2/3}+3)/16$ and that the central critical value corresponds to the values $H_1=2-3\cdot2^{-2/3}$, $H_2=-2+3\cdot2^{-2/3}$, $K=(12\cdot2^{2/3}-1)/16$.

\printbibliography

\end{document}